\documentclass[11pt, a4paper]{article} 

\usepackage[german, english]{babel}
\usepackage[utf8]{inputenc}
\usepackage[T1]{fontenc}
\usepackage{lmodern}
\usepackage{amsmath}    
\usepackage{amsthm}  
\usepackage{amsfonts}     
\usepackage{amssymb}
\usepackage{array}
\usepackage[affil-it]{authblk}
\usepackage{caption}
\usepackage{color}
\usepackage{courier}
\usepackage{diagbox}
\usepackage{dsfont}
\usepackage{enumerate}			 
\usepackage{epsfig}
\usepackage{fancyhdr}
\usepackage{float}
\usepackage{geometry}
\usepackage{graphicx}   
\usepackage{cite}
\usepackage{listings}        
\usepackage{lastpage} 
\usepackage{mdframed}
\usepackage{multicol}
\usepackage{multirow}
\usepackage{newclude}
\usepackage{placeins}
\usepackage{scalerel}
\usepackage{setspace}
\usepackage{shuffle}
\usepackage{slashed}
\usepackage{subcaption}
\usepackage[tc]{titlepic}
\usepackage{titlesec}
\usepackage{url}
\usepackage{varwidth}
\usepackage{verbatim}
\usepackage{wrapfig}
\usepackage{xcolor}
\usepackage{CJK}
\usepackage{hyperref}

\newcommand*{\fermion}{ \includegraphics[height=0.5\baselineskip]{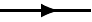} }
\newcommand*{\photon}{ \includegraphics[height=0.5\baselineskip]{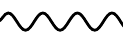} }
\newcommand*{\scalar}{ \includegraphics[height=0.5\baselineskip]{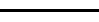} }
\newcommand*{\dashed}{ \includegraphics[height=0.5\baselineskip]{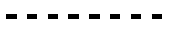} }
\newcommand*{\dotted}{ \includegraphics[height=0.5\baselineskip]{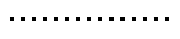} }

\newcommand{\dslash}{\mathbin{
  \mathchoice{/\mkern-6mu/}% \displaystyle
    {/\mkern-6mu/}% \textstyle
    {/\mkern-5mu/}% \scriptstyle
    {/\mkern-5mu/}}}% \scriptscriptstyle

\newcommand{\con}[2]{#1\!\dslash\!#2}
\newcommand{\cnh}[2]{#1/#2}

\newcommand{\nquad}{\!\!\!\!}
\newcommand{\nqquad}{\!\!\!\!\!\!\!\!}

\newcommand*{\defeq}{\mathrel{\vcenter{\baselineskip0.5ex \lineskiplimit0pt
                     \hbox{\scriptsize.}\hbox{\scriptsize.}}}%
                     =}
\newcommand*{\bdefeq}{= %
		\mathrel{\vcenter{\baselineskip0.5ex \lineskiplimit0pt
                     \hbox{\scriptsize.}\hbox{\scriptsize.}}}}
                     
\DeclareMathOperator{\adj}{adj}
\DeclareMathOperator{\id}{id}
\DeclareMathOperator{\Li}{Li}

\DeclareMathOperator{\sgn}{sgn} % signum
\DeclareMathOperator{\tr}{tr}

\newcolumntype{L}[1]{>{\raggedright\let\newline\\\arraybackslash\hspace{0pt}}m{#1}}
\newcolumntype{C}[1]{>{\centering\let\newline\\\arraybackslash\hspace{0pt}}m{#1}}
\newcolumntype{R}[1]{>{\raggedleft\let\newline\\\arraybackslash\hspace{0pt}}m{#1}}

\makeatletter
\newsavebox{\@brx}
\newcommand{\llangle}[1][]{\savebox{\@brx}{\(\m@th{#1\langle}\)}%
	\mathopen{\copy\@brx\kern-0.5\wd\@brx\usebox{\@brx}}}
\newcommand{\rrangle}[1][]{\savebox{\@brx}{\(\m@th{#1\rangle}\)}%
	\mathclose{\copy\@brx\kern-0.5\wd\@brx\usebox{\@brx}}}
\makeatother

\newcommand{\al}[2][]{\begin{align#1}#2\end{align#1}}

\newcommand{\sa}{\mathsf{a}}
\newcommand{\ssb}{\mathsf{b}}
\newcommand{\ssc}{\mathsf{c}}
\newcommand{\sd}{\mathsf{d}}

\newcommand{\su}{\mathsf{u}}
\newcommand{\sv}{\mathsf{v}}
\newcommand{\sw}{\mathsf{w}}
\newcommand{\sx}{\mathsf{x}}
\newcommand{\sy}{\mathsf{y}}

\newcommand{\sA}{\mathsf{A}}

\newcommand{\sP}{\mathsf{P}}

\newcommand{\cB}{\mathcal{B}}
\newcommand{\cC}{\mathcal{C}}
\newcommand{\cD}{\mathcal{D}}
\newcommand{\cE}{\mathcal{E}}

\newcommand{\cO}{\mathcal{O}}
\newcommand{\cP}{\mathcal{P}}

\newcommand{\cT}{\mathcal{T}}

\newcommand{\rd}{\mathrm{d}}
\newcommand{\ro}{\mathrm{o}}

\newcommand{\NN}{\mathbb{N}}
\newcommand{\ZZ}{\mathbb{Z}}

\newcommand{\RR}{\mathbb{R}}

\newtheorem{theorem}{Theorem}[section]
\newtheorem{example}[theorem]{Example}
\newtheorem{lemma}[theorem]{Lemma}
\newtheorem{definition}[theorem]{Definition}
\newtheorem{proposition}[theorem]{Proposition}

\newtheorem{remark}[theorem]{Remark}

\newcommand{\bp}[3][\Gamma]{ \beta_{#1}^{(#2|#3)} }
\newcommand{\cp}[3][\Gamma]{ \chi_{#1}^{(#2|#3)} }

\newcommand{\CP}[3][\Gamma]{ \mathrm{X}_{#1}^{#2, #3} }

\newcommand{\kp}[1][\Gamma]{ \Psi_{#1} }
\newcommand{\ssp}[1][\Gamma]{ \Phi_{#1} }
\newcommand{\vssp}[1][\Gamma]{ \varphi_{#1} }

\newcommand{\chGamma}[1][\gamma]{ \cnh{\Gamma}{#1} }

\newcommand{\reffig}[1]{fig. \ref{#1}}
\newcommand{\refeq}[1]{eq. (\ref{#1})}

\let\svthefootnote\thefootnote
\newcommand\blankfootnote[1]{%
  \let\thefootnote\relax\footnotetext{#1}%
  \let\thefootnote\svthefootnote%
}

\begin{document}

% Titlepage
	\title{\textbf{Dodgson polynomial identities}}
	\date{}
	\author{Marcel Golz}
	\affil{Institut f\"ur Physik, Humboldt-Universit\"at zu Berlin\\ Newtonstra\ss e 15, D-12489 Berlin, Germany}
	\maketitle
	\thispagestyle{empty}

% Abstract
	\begin{abstract}
		Dodgson polynomials appear in Schwinger parametric Feynman integrals and are closely related to the well known Kirchhoff (or first Symanzik) polynomial. In this article a new combinatorial interpretation and a generalisation of Dodgson polynomials are provided. This leads to two new identities that relate large sums of products of Dodgson polynomials to a much simpler expression involving powers of the Kirchhoff polynomial. These identities can be applied to the parametric integrand for quantum electrodynamics, simplifying it significantly. This is worked out here in detail on the example of superficially renormalised photon propagator Feynman graphs, but works much more generally.
	\end{abstract}
	\vspace{5mm}
	
% table of contents
	\setcounter{tocdepth}{2}
	\tableofcontents

% email footnote
	\blankfootnote{\texttt{mgolz@physik.hu-berlin.de}}

\newpage
%1
% INTRO
% INTRODUCTION
\section{Introduction}

	% Background
	\subsection{Background}
		
		Perturbative quantum field theory is the standard framework used by particle physicists to predict and explain high-energy experiments, e.g. at modern colliders like the LHC. This necessitates the computation of a large number of complicated integrals. These Feynman integrals grow quickly in number and complexity, so on the one hand one wants to find methods to compute them as efficiently as possible, and on the other hand one looks for hidden structures that reduce the amount of necessary computations.
		
		To that end, the Schwinger parametric representation of Feynman integrals has proved to be very useful in recent years. It was already known in the early days of quantum field theory \cite{ Nambu_1957_ParaRep, Nakanishi1957, Nakanishi1961, Symanzik1958, Kinoshita1962, BergereZuber_1974_ParaRenorm, Bergere1976 }, but fell somewhat out of favour, since it was not as suitable for direct integration as other versions of Feynman integrals. This problem was rectified when, building on the connections to algebraic geometry found in \cite{BlochEsnaultKreimer_2006_Motives}, an algorithm for the systematic integration of parametric Feynman integrals was developed \cite{Brown_2009_Massless, Brown_2010_Periods} and subsequently implemented in computer algebra \cite{Panzer_2014_Algorithms}.
		
		The renewed interest in parametric Feynman integrals has already yielded many interesting results \cite{BrownSchnetz_2012_K3, BrownYeats_2011_SpanningForest, BrownSchnetzYeats_2014_C2, KreimerSarsSuijlekom_2013_QuantGauge, BrownKreimer_2013_AnglesScales, KompanietsPanzer_2016, BBKP_ParaAnn_2018, Bloch_Kreimer_2015 }. However, those have all been confined to scalar quantum field theories. In that case the (unrenormalised) integral is simply of the form
		\al{
			\phi_{\Gamma} = \int_{\RR_+^{|E_{\Gamma}|}} \rd\alpha_1\dotsm \rd\alpha_{E_{\Gamma}} \frac{\exp\Big( -\frac{\ssp}{\kp} \Big)}{\kp^2},
		}
		where $\Gamma$ is a Feynman graph and $\kp$, $\ssp$ are two homogenous polynomials that will be discussed extensively below.

		In gauge theories this becomes much more complicated and until recently the integrand could only be expressed in terms of complicated derivatives of the scalar integrand \cite{CvitanovicKinoshita_1974_FRPS, KreimerSarsSuijlekom_2013_QuantGauge}. For quantum electrodynamics the combinatorics of these derivatives have been analysed in \cite{Golz_2017_CyclePol} and it was found that they can be expressed explicitly in terms of graph polynomials similar to $\kp$ and $\ssp$. The other complication of QED, the tensor structure consisting of products of Dirac matrices, was dealt with in \cite{Golz_2017_Traces}. Combining these results yields an (unrenormalised, massless) parametric Feynman integral for QED that is of the form
		\al{
			\phi_{\Gamma} = \int_{\RR_+^{|E_{\Gamma}|}} \rd\alpha_1\dotsm \rd\alpha_{E_{\Gamma}} \frac{\exp\Big( -\frac{\ssp}{\kp} \Big)}{ \kp^{2+h_1(\Gamma)} } 
																					\sum_{l=0}^{h_1(\Gamma)} \frac{ I_{\Gamma}^{(l)} }{\kp^l} \label{eq_int_unren}
		}		
		where each $I_{\Gamma}^{(l)}$ is essentially (cf. \refeq{eq_sumdetail} and sec. \ref{sec_application}) just a sum over certain subsets of chord diagrams $D$,
		\al{
			\sum_D (-2)^{\tilde c(D)} X_D,
		}
		where $X_D$ is a product of the polynomials from \cite{Golz_2017_CyclePol} and $\tilde c(D)$ is an integer determined by the combinatorial properties of $D$. In our main results, theorems \ref{theo_main} and \ref{theo_main_2}, we prove that the sums in $I_{\Gamma}^{(0)}$ and $I_{\Gamma}^{(1)}$ are equal to a simpler sum of the form
		\al{
			2^{-k}\sum_{l=1}^{h_1(\Gamma)} (-\kp)^{h_1(\Gamma)-l+k} (l+1)!\ Z_{\Gamma}^k\big|_l\qquad\qquad\text{for } k=0,1,	\label{eq_result_intro}
		}
		where $Z_{\Gamma}^k\big|_l$ is defined in sec. \ref{sec_defpp}. This leads to cancellations of Kirchhoff polynomials $\kp$ in \refeq{eq_int_unren}, significantly simplifying the integrand. On the concrete example of a massless photon propagator graph in Feynman gauge we show that the cases of $k=0,1$ suffice to express the superficially renormalised integral with a simple entirely scalar integrand.
		
		% generalisations
		\paragraph{Generalisations and extensions.}
			The results of this article are not just applicable to this rather specific photon propagator. In a general gauge one gets another sum and our results apply to each summand (see \refeq{eq_gengauge}). The same holds for the inclusion of masses, if one assumes quite reasonably that all fermion masses are identical. The parametric renormalisation of massive integrals is much more cumbersome than the simple renormalisation procedure that we employ in section \ref{subsec_structure}, but in principle not a problem \cite{BrownKreimer_2013_AnglesScales}. 
			
			For a fermion propagator and a vertex with one external momentum set to zero the differences are basically just a few different factors in the computations of section \ref{subsec_structure} (e.g. the fermion propagator would be proportional to $\slashed q$ rather than $q^2q^{\mu\nu}-q^{\mu}q^{\nu}$). The step to the full vertex function is more complicated and needs more attention in future work. Much of this article and especially section \ref{sec_inc} is based on the assumption that the polynomial $\ssp$ factorises into $q^2\vssp$ with a $q$-independent $\vssp$. It will need to be seen how much the results have to be modified if one has instead a polynomial $\ssp = q_1^2 \vssp[\Gamma,1] + q_2^2 \vssp[\Gamma,2] + (q_1+q_2)^2 \vssp[\Gamma,3]$. 

			Finally, in order to include subdivergences one also needs to understand the higher order terms $I_{\Gamma}^{(k)}$ with $k\geq 2$. In \refeq{eq_result_intro} we already suggest what this should look like, although it is not yet entirely clear how the $Z_{\Gamma}^k\big|_l$ have to be defined for $k\geq 2$.\\[3mm]

%	We continue the introduction with a review of the various types of graph polynomials we will encounter, as well as the chord diagrams that encode the results of the Dirac matrix contractions. In section 

	% Graph polynomials
	\subsection{Graph polynomials}
		
	A graph $G$ is an ordered pair $(V_G, E_G)$ of the set of \textit{vertices} $V_G = \{ v_1, \dotsc, v_{|V_G|}\}$ and the set of \textit{edges} $E_G = \{ e_1, \dotsc, e_{|E_G|}\}$, together with a map $\partial: E_G \to V_G \times V_G$. We assume that $G$ is connected and assign to each edge $e\in E_G$ a direction by specifying an ordered pair $\partial(e) = ( \partial_-(e), \partial_+(e) )$, where the vertex $\partial_-(e)\in V_G$ is called start or initial vertex while $\partial_+(e)\in V_G$ is called target or final vertex. In a common abuse of notation subgraphs $g\subseteq G$ are identified with their edge set $E_g\subseteq E_G$. In the rare cases in which the edge set does not uniquely identify the subgraph, i.e. when $g$ contains isolated vertices without incident edges, it will be mentioned explicitly. The number of independent cycles (loops, in physics nomenclature) is denoted $h_1(G)$, which is the first Betti number of the graph.
	
	Graph polynomials are polynomial valued invariants of a graph. The polynomials that we are interested in all have in common that their variables are the Schwinger parameters $\alpha = (\alpha_e)_{e \in E_G}$ assigned to the edges of a graph (which distinguishes them from other famous graph polynomials like the Tutte polynomial \cite{Tutte_1954, Tutte_2004} and its various specialisations like the chromatic polynomial \cite{Birkhoff_1912, Whitney_1932_coloring}). In the following we briefly introduce and review some properties of six such graph polynomials that appear in Feynman integrals.

		% KP
		\subsubsection{Kirchhoff and Symanzik}
				
		A tree $T$ is a graph that is connected and simply connected, i.e. it has no cycles. A disjoint union of trees $F=\sqcup_{i=1}^kT_i$ is called a $k$-forest, such that a tree is a $1$-forest. If all vertices of $G$ are contained  in such a subgraph $T$ or $F$, then it is called a spanning tree or spanning forest of $G$ and we denote with $\cT_G^{[k]}$ the set of all such spanning $k$-forests.
	
	The Kirchhoff polynomial, which is especially in the physics literature also often called the first Symanzik polynomial, is then defined as
	\al{
		\kp[G](\alpha) \defeq  \sum_{T\in \cT_G^{[1]}} \prod_{e\notin T} \alpha_e.
	}
	It has been known for a very long time and was first introduced by Kirchhoff in his study of electrical circuits \cite{Kirchhoff_1847_Kirchhoff}. In the 1950s it was then rediscovered in quantum field theory \cite{Symanzik1958}. We will often make use of the abbreviation
	\al{
		\alpha_S \defeq \prod_{e \in S}\alpha_e 
	}
	for any edge subset $S\subset E_G$, such that $\kp[G] = \sum_T \alpha_{E_G \setminus T}$. The Kirchhoff polynomial is homogeneous of degree $h_1(G)$ in $\alpha$ and linear in each $\alpha_e$. Moreover, it also satisfies the famous contraction-deletion relation\footnote{Note that we use the double-slash to denote contraction of an edge subset, as opposed to contraction of a subgraph $\chGamma$ in the Hopf algebra of Feynman graphs. The two notions differ if the subgraph in question is a propagator Feynman graph, but in this article we will not encounter this problem.}
	\al{
		\kp[G] = \kp[G\dslash e] + \alpha_e\kp[G\setminus e].
	}
	This means in particular that $\kp[G\dslash e] = \left.\kp[G]\right|_{\alpha_e=0}$ and $\kp[G\setminus e] = \partial_e \kp[G]$, where $\partial_e$ denotes the partial derivative w.r.t. $\alpha_e$. The definition of the Kirchhoff polynomial is commonly generalised to disjoint unions of graphs $G=\sqcup_i G_i$ (i.e. graphs with multiple connected components) via
	\al{
		\kp[G] \defeq \prod_i \kp[G_i]. \label{eq_kp_disj}
	}
	Note that due to this definition one needs to exclude bridges from the contraction-deletion relation, since $\partial_e \kp[G]=0$ for any bridge $e$, whereas this definition gives $\kp[G\setminus e]$ as the product of the polynomials of its two connected components.\\

	Many properties of graphs can be captured by matrices and we discuss here some of the well known relations between graphs, matrices and the Kirchhoff polynomial. The incidence matrix $I$ is an $|E_G| \times |V_G|$ matrix whose entries are defined as
	\al{
		I_{ev} \defeq \begin{cases} \hspace{3mm} \pm1 &\text{ if $v = \partial_{\pm}(e)$},\\
						    \hspace{3mm} 0 &\text{ if $e$ is not incident to $v$}.  \end{cases}
	}
	The Laplacian matrix $L$ is defined as the difference of the degree and adjacency matrices of a graph. Since we will not need either of those two going forward we instead use a well known identity to define the Laplacian as the product of the incidence matrix and its transpose,
	\begin{align}
		L \defeq I^TI. 
	\end{align}
	Instead of the full matrices we will actually always need the smaller matrices in which one column (of $I$) or one column and one row (of $L$) corresponding to an arbitrarily chosen vertex $v_0$ of $G$ are deleted. From now on we use $I'$ and $L'$ for these $|E_G| \times |V_G|-1$ and $|V_G|-1 \times |V_G|-1$ matrices, called reduced incidence and reduced Laplacian matrix.
	
	Finally, let $A$ be the diagonal $|E_G|\times |E_G|$ matrix with entries $A_{ij} \defeq \delta_{ij} \alpha_{e_i}$. With this setup the well known \textit{Matrix-Tree-Theorem} \cite{Chaiken_1978} tells us that
	\al{
		\kp[G] = \alpha_{E_G} \det( {I'}^TA^{-1}I' ). 
	}
	Note that here we have the inverse $A^{-1}$, with entries $A_{ij}^{-1} = \delta_{ij} \alpha_{e_i}^{-1}$. We call the matrix in that determinant the weighted reduced Laplacian and denote it with $\tilde L' = {I'}^TA^{-1}{I'}$.
	
	\begin{remark}
		The polynomial $\kp[G]^* = \det( {I'}^TAI' )$ is sometimes called dual Kirchhoff polynomial. If $G$ is planar then it is the Kirchhoff polynomial of its planar dual graph $G^*$.
	\end{remark}

	\noindent Often $I'$ and $A$ are arranged in a block matrix
	\al{
		M \defeq  \begin{pmatrix} A & I' \hphantom{x}\\[2mm] -{I'}^T & 0 \hphantom{x} \end{pmatrix}. \label{eq_graphmatrix}
	}
	This is called the expanded Laplacian or graph matrix of $G$ \cite{Brown_2010_Periods, BlochEsnaultKreimer_2006_Motives}, and with the block matrix determinant identity
	\al{
		\det \begin{pmatrix} S & T\\[1mm] U & V \end{pmatrix} = \det(S)\det(V - US^{-1}T) \label{eq_blockID}
	}
	one sees that
	\al{
		\det(M) = \det(A)\det( {I'}^TA^{-1}I' ) = \kp[G].  \label{eq_detgmkp}
	}
	
	 \begin{example}\label{ex_kp1}
	 	Let $G$ be the wheel with three spokes depicted on the left of \reffig{01_WheelFeynman}. It has 16 spanning trees and the Kirchhoff polynomial is
		\al{\nonumber 
			\kp[G] &= \alpha_1\alpha_2\alpha_4 + \alpha_1\alpha_2\alpha_5 + \alpha_1\alpha_3\alpha_4 + \alpha_1\alpha_3\alpha_5
				+ \alpha_1\alpha_2\alpha_6 + \alpha_1\alpha_3\alpha_6 + \alpha_1\alpha_4\alpha_6 + \alpha_1\alpha_5\alpha_6\\[2mm]
				&+ \alpha_2\alpha_3\alpha_4 + \alpha_2\alpha_3\alpha_5 + \alpha_2\alpha_4\alpha_5 + \alpha_3\alpha_4\alpha_5
				+ \alpha_2\alpha_3\alpha_6 + \alpha_2\alpha_4\alpha_6 + \alpha_3\alpha_5\alpha_6  + \alpha_4\alpha_5\alpha_6.
			}
	 \end{example}

	\begin{figure}[h]
		\begin{center}
				\includegraphics[width=0.3\textwidth]{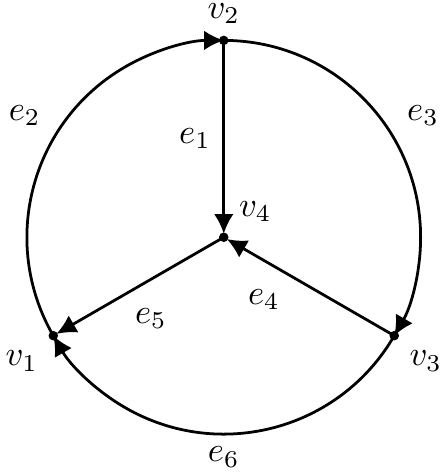}\hspace{15mm}
				\includegraphics[width=0.52\textwidth]{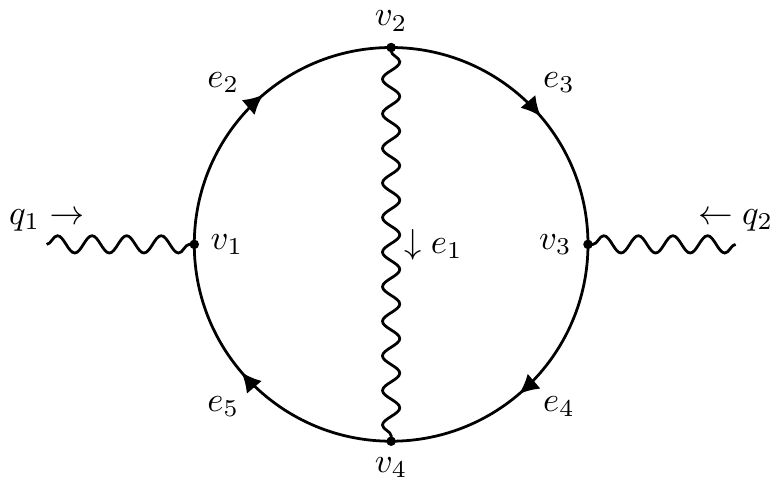}
		\end{center}
		\caption{The wheel with three spokes $G=W\!S_3=K_4$ and a QED Feynman graph $\Gamma$ that corresponds to the wheel with edge $e_6$ cut to become an external photon edge.}
		\label{01_WheelFeynman}
	\end{figure}
	
	Unlike the Kirchhoff polynomial, the second Symanzik polynomial is not defined for generic graphs but only for Feynman graphs, which carry additional information. Feynman graphs can have different types of edges, like photons ($\photon$) and fermions ($\fermion$) in the graph on the r.h.s. of \reffig{01_WheelFeynman}. Moreover, Feynman graphs have so-called external edges, i.e. half edges incident to only a single vertex. One associates to the external edges external momenta $q_i$, that are real euclidean vectors $q_i\in \RR^4$ for this article, but may also be minkowskian, $D$-dimensional, or complex, depending on context. To distinguish them from generic graphs $G$ we use $\Gamma$ for Feynman graphs. The second Symanzik polynomial is then defined as \cite{Symanzik1958}
	\al{
		\ssp(\alpha, q) \defeq \sum_{(T_1,T_2) \in \cT_{\Gamma}^{[2]}} \nquad s(q,T_1,T_2) \prod_{e\notin T_1\cup T_2} \alpha_e 
	}
	where one sums over spanning $2$-forests. The function $s$ is the square of the momentum flow between the two trees, i.e. the sum of all external momenta entering either tree (which is the same for both trees due to momentum conservation). 
	
	If, as in \reffig{01_WheelFeynman} for example, there are only two (non-zero) external momenta, such that $q_1 = -q_2 \equiv q$ by momentum conservation, then the second Symanzik polynomial factorises and we write
	\al{
		\ssp = q^2\vssp.\label{eq_ssp_fac}
	}
	We focus on this case for this article. Note that $\vssp$ is also a Kirchhoff polynomial, namely that of the graph $\Gamma^{\bullet}$, which results from adding the external edge between the two external vertices and then contracting it.
	
	The second Symanzik polynomial can also be expressed in terms of matrices. When deriving parametric Feynman integrals it appears in the form of the inverse Laplacian $\tilde {L'}^{-1}$ multiplied from both sides with vectors collecting all external momenta. Using cofactors to invert the matrix and expanding the matrix products as sums this yields
	\al{
		\ssp = \alpha_{E_{\Gamma}} \!\! \sum_{v_1,v_2 \in V_{\Gamma}'}\nquad q_{v_1}\!\! \cdot q_{v_2} (-1)^{v_1+v_2}\det( \tilde {L'}^{\{v_1\}}_{\{v_2\}} ),\label{eq_ssp_matrix}
	}
	where $V_{\Gamma}' = V_{\Gamma} \setminus \{v_0\}$ is the set of all vertices except the one whose row and column was removed from all matrices to get their reduced versions.	

	% EXAMPLE 
	\begin{example}
		Let $\Gamma$ be the Feynman graph on the r.h.s of \reffig{01_WheelFeynman}. Its Kirchhoff polynomial is just
		\al{\nonumber
			\kp &= \alpha_1\alpha_2 + \alpha_1\alpha_3 + \alpha_1\alpha_4 + \alpha_1\alpha_5
							+\alpha_2\alpha_3 + \alpha_2\alpha_4 + \alpha_3\alpha_5 + \alpha_4\alpha_5\\
						&=  (\alpha_2+\alpha_5)(\alpha_3+\alpha_4) + \alpha_1(\alpha_2+\alpha_3+\alpha_4+\alpha_5),
		}
		which is the derivative w.r.t. $\alpha_6$ of the Kirchhoff polynomial from example \ref{ex_kp1}.
		
		There are a total of 10 spanning 2-forests, but not all of them contribute to the second Symanzik polynomial. Consider the spanning 2-forest with $T_1 = \{e_2,e_5\}$ and $T_2$ just the isolated vertex $v_3$ without edges. The external momentum $q_1$ enters $T_1$ in the vertex $v_1$ and $q_2$ (which has to be $-q_1$ due to momentum conservation) enters $T_2$ in $v_3$. Hence, the corresponding monomial is $q_1^2\alpha_1\alpha_3\alpha_4$. An example of a forest that does not contribute is $T_1=\{e_2,e_3\}$ and $T_2$ just the vertex $v_4$. The external momenta entering $T_1$ in $v_1$ and $v_3$ add up to $0$, whereas $T_2$ is not adjacent to any external edges at all. Hence $s(q_1,q_2,T_1,T_2) = 0$ in this case. Overall, 8 of the 10 forests contribute to yield the second Symanzik polynomial
		\al{
			\ssp = q^2\vssp & =  q^2\big( \alpha_2\alpha_5(\alpha_1+\alpha_3+\alpha_4)
							+\alpha_3\alpha_4(\alpha_1+\alpha_2+\alpha_5)+ \alpha_1\alpha_2\alpha_4+\alpha_1\alpha_3\alpha_5\big).
		}
		Expanding the polynomial one sees that 
		\al{
			\vssp = \alpha_1\alpha_2\alpha_4 + \alpha_1\alpha_2\alpha_5 + \alpha_1\alpha_3\alpha_4 + \alpha_1\alpha_3\alpha_5
				 + \alpha_2\alpha_3\alpha_4 + \alpha_2\alpha_3\alpha_5 + \alpha_2\alpha_4\alpha_5 + \alpha_3\alpha_4\alpha_5
		}
		is indeed $\kp[\Gamma^{\bullet}] = \kp[\con{G}{e_6}] = \kp[G]\big|_{\alpha_6=0}$, where $G$ is the wheel from example \ref{ex_kp1}.
	\end{example}

	% CP + BP
	\subsubsection{Bonds and cycles}
		
	A \textit{bond} $B\subset G$ is a minimal subgraph $G$ such that $G\setminus B$ has exactly two connected components.  A simple \textit{cycle} $C\subset G$ is a subgraph of $G$ that is 2-regular, i.e. all vertices have exactly two edges incident to it, and it has only one connected component. The sets of bonds and simple cycles of a graph $G$ are denoted $\cB_G$ and $\cC_G^{[1]}$.
		
	In \cite{Golz_2017_CyclePol} two polynomials based on these types of subgraphs were defined and it was shown that they can be used to express the Schwinger parametric integrand in quantum electrodynamics without derivatives. The basic bond polynomial and cycle polynomial are
	\al{
		\beta_G(\alpha,\xi) &\defeq \sum_{B \in \cB_G} \left(\sum_{e \in B} \ro_B(e) \xi_e  \right)^2 \alpha_B\kp[G\setminus B],\\[3mm]
		\chi_G(\alpha,\xi) &\defeq \sum_{C \in \cC_G^{[1]}} \left(\sum_{e \in C} \ro_C(e) \xi_e  \right)^2   \kp[G\dslash C],
	}
	where $\xi = (\xi_1, \dotsc, \xi_{|E_G|})$ are formal parameters assigned to each edge (later interpreted as auxiliary momenta, i.e. euclidean 4-vectors), and $\ro_B(e), \ro_C(e) \in \{ 0, \pm 1\}$ are the signs of the relative orientations of $e$ w.r.t. some arbitrarily chosen orientation of the bond or cycle. Note that this choice does not influence the sign of the polynomials, since these orientations only appear within the square. Below we will abbreviate products of such signs as $\ro_C(e_1)\ro_C(e_2) = \ro_C(e_1,e_2)$. 
		
	 Via \refeq{eq_kp_disj} -- the Kirchhoff polynomial definition for disconnected graphs -- these definitions extend to disconnected $G$ as well. For Feynman graphs the bond polynomial is closely related to the second Symanzik polynomial. In fact, $\ssp(\alpha, q)$ is simply the evaluation of $\beta_{\Gamma}(\alpha,\xi)$, where one sends $\xi_e\to q$ for each $e$ in some arbitrary path between the external vertices and all others to $0$ (and if there are $n>2$ external vertices, then one does this for $n-1$ pairs of external vertices to get the correct linear combinations $\xi_e\to \sum \pm q_i$).\\

	From these two polynomials we derive two families of polynomials that we will from now on mostly mean when speaking of cycle or bond polynomials:
	\al{
		\bp[G]{e_i}{e_j}(\alpha) &\defeq \frac{1}{2}\frac{\partial^2\beta_G}{\partial \xi_i\partial \xi_j} = \sum_{B \in \cB_G} \ro_B(e_i,e_j) \alpha_B \kp[G\setminus B]\\[3mm]
		\cp[G]{e_i}{e_j}(\alpha) &\defeq \frac{1}{2}\frac{\partial^2\chi_G}{\partial \xi_i\partial \xi_j} = \sum_{C \in \cC_G^{[1]}} \ro_{C}(e_i,e_j) \kp[G\dslash C]				\label{eq_defcycpol}
	}
	
	Cycle and bond polynomials inherit many useful properties from the Kirchhoff polynomial. They are clearly still linear in each $\alpha_e$ and homogenous of degree $h_1(G)+1$ (for $\beta_G$, $\bp[G]{e_i}{e_j}$) and $h_1(G)-1$ (for $\chi_G$, $\cp[G]{e_i}{e_j}$). They also satisfy the contraction-deletion relations and the following three useful identities (proposition 2.8 and lemmata 2.9, 2.10, 2.11 in \cite{Golz_2017_CyclePol}):
	\al{
		\cp[G]{e}{e} &= \kp[G\setminus e] = \frac{\partial}{\partial \alpha_e} \kp[G] \hspace{20mm} \text{ if $e$ is not a bridge.} \\[3mm]
		\bp[G]{e}{e} &= \alpha_e\kp[G\dslash e] = \alpha_e\left.\kp[G]\right|_{\alpha_e=0} \hspace{8mm} \text{ if $e$ is not a self-cycle.}\\[5mm]
		\bp[G]{e}{e'} &= -\alpha_e\alpha_{e'} \cp[G]{e}{e'} \hspace{24mm} \text{ if } e \neq e'.	\label{eq_lemma_bpcp}
	}
	We also need the polynomial
	\al{
		\CP{e}{\mu}(\alpha,\xi) 
		\defeq \frac{1}{2\alpha_e}\frac{\partial}{\partial \xi_{e,\mu}} \beta_{\Gamma}(\alpha, \xi)
		= \alpha_e^{-1} \sum_{e' \in E_{\Gamma}} \xi_{e'}^{\mu} \bp{e}{e'}.	\label{eq_CPdef}
	}
	More specifically, we need its evaluation $\xi \to q$ for the case of a single external momentum,
	\al{
		\CP{e}{\mu}(\alpha,q) = q^{\mu} x_{\Gamma}^e(\alpha),	\label{eq_CP_eval_def}
	}
	which factorises similarly to the second Symanzik polynomial in \refeq{eq_ssp_fac}.
					
	% DP
	\subsubsection{Dodgson and spanning forests}	

		In \refeq{eq_detgmkp} we have seen that the Kirchhoff polynomial can be written as the determinant of the graph matrix $M$. Motivated by this one considers minors of the graph matrix, i.e. determinants
		\al{
			\kp[G]^{I,J} \defeq \det\big( M_I^J\big),
		}
		where the edge subsets $I,J \subset E_G$ with $|I|=|J|$ in the subscript and superscript denote deletion of all rows and columns indexed by edges in the respective set. In general one often uses a third index set $K$ for $\kp[G,K]^{I,J}$ and sets $\alpha_e=0$ for all $e\in K$, but here we always have $K=\emptyset$. Note that $\kp[G]^{I,J}$ is only well-defined up to an overall sign since a different ordering of the rows and columns in the graph matrix may change the sign of the determinant. This will be discussed further below, but for now we just fix one such ordering.

	The $\kp[G]^{I,J}$ are called Dodgson polynomials and appeared already in \cite{BlochEsnaultKreimer_2006_Motives}. They were first named and systematically studied by Francis Brown in \cite{Brown_2010_Periods}. In the following we discuss some notable properties.
	\paragraph{Passing to a minor.}
			For all $A, B \subset E_G$
			\al{
				\Psi_{G\setminus A\dslash B,K}^{I,J} = \Psi_{G,K\cup B}^{I\cup A,J\cup A},
			}
			which justifies our setting $K=\emptyset$.
	\paragraph{Determinant identities.} 
			Let $\adj(M)[I,J]$ be the restriction of the adjugate matrix of $M$ to rows and columns indexed by $I$ and $J$. Based on the Desnanot-Jacobi identity \cite{Dodgson_1866_CondDet}
			\al{
				\det(\adj(M)[I,J]) = \det(M)^{|I|-1}\det(M^I_J)\label{eq_DJ}
			}
			for determinants one finds identities of the type
			\al{
				\Psi_G^{\{i_1\},\{i_3\}}\Psi_G^{\{i_2\},\{i_4\}} - \Psi_G^{\{i_1\},\{i_4\}}\Psi_G^{\{i_2\},\{i_3\}} = \kp[G] \Psi_G^{\{i_1,i_2\},\{i_3,i_4\}}.\label{eq_DodgsonID}
			}
			This case ($|I|=2=|J|$) is also called \textit{Dodgson identity}\footnote{Somewhat confusingly, it is also occasionally called \textit{Lewis Carroll identity} and both names are sometimes used to refer to the determinant identity \refeq{eq_DJ} \cite{Bressoud}. Here we follow the conventions of \cite{Brown_2010_Periods}.} and its generalisations are the crucial tool that we will work with below.
	\paragraph{Combinatoric interpretation.} 
			In the case of $I \cap J = \emptyset = K$ the combinatoric interpretation for Dodgson polynomials given by Brown in \cite[Prop. 23]{Brown_2010_Periods} simplifies to
			\al{
				\Psi_G^{I,J} = \sum_{T\subset E_G\setminus(I\cup J)} \pm \prod_{e\notin T} \alpha_e\label{eq_DodgsonCombInt}
			}
			where the sum is over edge subsets $T$ that are simultaneously spanning trees of $\con{(G\setminus I)}{J}$ and $\con{(G\setminus J)}{I}$. A criterion for two monomials in this sum to have the same or opposite signs is given in \cite[Corollary 17]{BrownYeats_2011_SpanningForest}. Moreover, if $I$ and $J$ do intersect, then
			\al{
				\Psi_G^{I,J} = \Psi_{G\setminus (I\cap J)}^{I\setminus J,J\setminus I}
			}
			if $G\setminus (I\cap J)$ is still connected and zero otherwise. In particular, $\Psi_G^{\{e\},\{e\}} = \kp[G\setminus e]$.\\[3mm]

	% spanning forest
	Finally, the last well-known graph polynomial that we will need is the spanning forest polynomial \cite[Def. 9]{BrownYeats_2011_SpanningForest}
	\al{
		\Phi_G^P = \sum_{F = T_1\sqcup \dotsm \sqcup T_k \in \cT^{[k]}_P} \alpha_{E_G\setminus F},
	}
	where $P = P_1,\dotsc,P_k$ is a partition of vertices of $G$ and $\cT^{[k]}_P$ is the subset of spanning $k$-forests which have the vertices of $P_i$ contained in the tree $T_i$.
	
	Being a sum over spanning forests and denoted by the same letter it is no surprise to find that these polynomials are closely related to the second Symanzik polynomial. Consider the matrix expression for $\ssp$ from \refeq{eq_ssp_matrix}. The coefficients of a product $q_{v_1}\!\! \cdot q_{v_2}$ of external momenta are precisely the spanning forest polynomials $\ssp^{\{v_0\},\{v_1,v_2\}}$, such that
	\al{
		\ssp = \alpha_{E_{\Gamma}} \sum_{v_1,v_2 \in V_{\Gamma}'}\nquad q_{v_1}\!\! \cdot q_{v_2} (-1)^{v_1+v_2}\det( \tilde {L'}^{\{v_1\}}_{\{v_2\}} ) 
			= \sum_{v_1,v_2 \in V_{\Gamma}}\nquad q_{v_1}\!\! \cdot q_{v_2}\Phi_G^{\{v_0\},\{v_1,v_2\}}.			\label{eq_ssp_stmat}
	}
	Note that $\Phi_G^{\{v_0\},\{v_1,v_2\}}=0$ if either of the two vertices $v_1,v_2$ is equal to $v_0$. Hence, in this expression we can just sum over the entire vertex set and do not need to write $V_{\Gamma}'$.
	
	% Chord diagrams
	\subsection{Chord diagrams}
		
	Aside from Feynman graphs we need another very special kind of graph -- chord diagrams. They can be used to model the contraction and traces of Dirac matrices, which is why they appear in QED Feynman integrals. For proofs and a more in-depth discussion we refer to \cite{Golz_2017_Traces}.

	Classically, chord diagrams consist of a cycle on $2n$ vertices (the base) and $k\leq n$ additional edges that pairwise connect $2k$ of the vertices of that cycle (the chords), but here we need a slightly more general definition that allows for multiple base cycles.

	% DEF chord diagram
	\begin{definition}
		A chord diagram is a graph that consists of $\ell\geq 1$ cycles with $2n_1, \dotsc, 2n_{\ell}$ vertices and $k \leq \sum n_i\bdefeq N$ further edges, such that each vertex is at most 3-valent. 
	\end{definition}

	\begin{figure}[h]
		\begin{center}
			\includegraphics[width=0.95\textwidth]{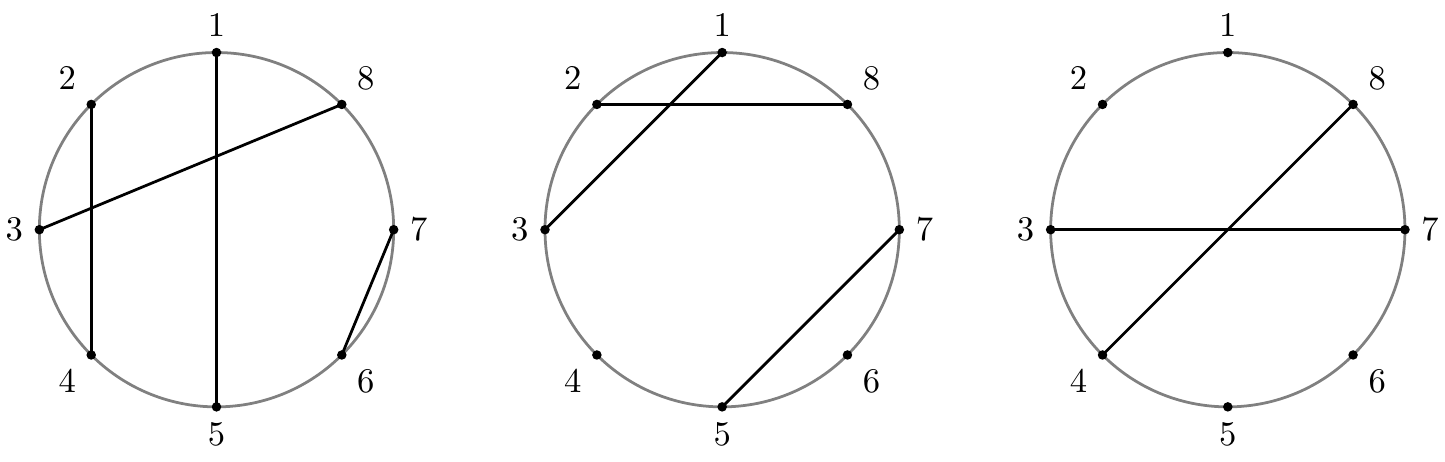}
			\caption{Three chord diagrams of order $n=4$ with one base cycle and four, three and two chords respectively.} 
			\label{img_chordexamples}
		\end{center}
	\end{figure}

	We denote with $\cD^n_k$ the set of all chord diagrams\footnote{Here we always mean labelled diagrams, i.e. two diagrams that are isomorphic as graphs but differ in the labelling of the vertices are viewed as different chord diagrams. In practice we will always either fix an arbitrary labelling $1,\dotsc, 2N$ that does not influence the result, or have a labelling fixed from context because the diagram is derived from a Feynman graph in a certain way.} with the respective number and size of base cycles and chords, determined by the $\ell$-tuple $n=(n_1,\dotsc, n_{\ell}) \in \NN_+^{\ell}$ and $1\leq k \leq N$. The set of 2-valent vertices of a chord diagram is denoted $V_D^{(2)}$ and we will often call them the ``free vertices'' of $D$.\\

	% Colours and cycles
	\subsubsection{Colours}\label{sec_colours}
	
	In addition to the distinction between base edges and chords we will need to introduce more properties to differentiate between certain types of edge subsets. This is achieved via colouring. For some finite set $K$ a map $\kappa: E_G\to K$ is called an edge $k$-colouring if for every vertex $v$ of $G$ all edges incident to it are assigned different colours, i.e. if $\kappa$ is injective in the neighbourhood of $v$. Chord diagrams $D$ admit an edge 3-coloring $\kappa: E_D \to \{0,1,2\}$ that assigns two alternating colours $1$ and $2$ to the edges of the base cycles and the third colour $0$ to all chords. There are $2^{\ell}$ possibilities of such a colouring corresponding to the exchange of colours $1$ and $2$ in some base cycles, so from now on we fix one such choice in all diagrams. In drawings we visualise the colours with different line types:
		\al[*]{
			0 \sim \dotted \hspace{2cm} 1 \sim \scalar \hspace{2cm} 2 \sim \dashed
		}
	
	Let $E_D^{i} = \kappa^{-1}(\{i\})$, $E_D^{ij} = \kappa^{-1}(\{i,j\})$ be the edge subsets consisting only of edges of the respective colour or colours. Each bicoloured edge subset can be decomposed into collections of cycles $\cC_D^{ij}$ and paths $\cP_D^{ij}$, where $\cP_D^{12} = \emptyset$, $|\cP_D^{01}| = |\cP_D^{02}|$, $|\cC_D^{12}| = \ell$, and we define $c_2(D) \defeq |\cC_D^{01}| + |\cC_D^{02}|$. The bicoloured paths between the 2-valent vertices of $D$ can be joined in their shared vertices to build tricoloured cycles, whose number we define to be $c_3(D)$. Often we are only interested in the total number of such coloured cycles, which we call $\tilde c(D) \defeq c_2(D) + c_3(D)$. Beware that this excludes the base cycles in $\cC_D^{12}$, which are counted separately by $\ell$.
	
	% FIGURE colour
	\begin{figure}[H] \begin{center}
		\includegraphics[width=0.45\textwidth]{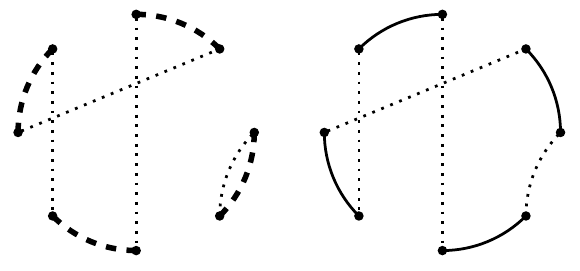}\hspace{10mm}
		\includegraphics[width=0.45\textwidth]{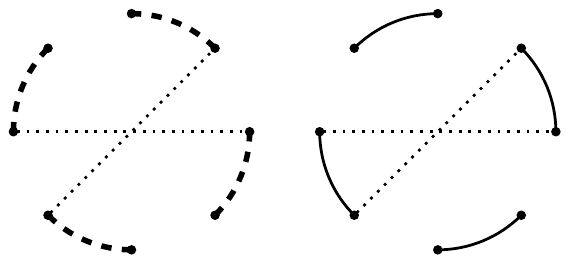}
		\caption{Colour decompositions of the left and right chord diagrams from \reffig{img_chordexamples}.}
		\label{img_dc}
	\end{center} \end{figure}
	
	% EXAMPLE	
	\begin{example}\label{04b_ex_threecol}
		Let $D_1, D_2, D_3$ be the three chord diagrams, left to right, from \reffig{img_chordexamples}. For $D_1$ and $D_3$ the bicoloured subsets are depicted in \reffig{img_dc}. Since there are no 2-valent vertices in $D_1$ all bicoloured components are cycles. Simply counting them in the drawing one finds
		\al[*]{
			c_2(D_1) = 3 \qquad c_3(D_1) = 0.
		}
		$D_3$, on the other hand, has four free vertices. We see that there is still a bicoloured cycle on the very r.h.s. of \reffig{img_dc}. The other bicoloured subsets are paths, four of them, which combine into a single tricoloured cycle, such that
		\al[*]{
			c_2(D_3) = 1 \qquad c_3(D_3)=1.
		}
		We leave it as an exercise for the reader to draw the bicoloured subsets of $D_2$, count the cycles, and confirm that also
		\al[*]{
			c_2(D_2) = 1 \qquad c_3(D_2)=1.
		}
	\end{example}
	
	Contracting each path in $\cP_D^{0i}$ to a single edge of colour $i$ maps the tricoloured cycles to a chord diagram $D_0 \in \cD_0^{n_0}$, for some suitable $n_0$ with $\sum n_{0,i} \leq N$, without any chords. Its base cycles correspond to the tricoloured cycles and its vertices are the 2-valent vertices of $D$. We call this projection map $\pi_0$.
	
	% FIGURE projection
	\begin{figure}[H]
		\begin{center}
			\includegraphics[width=0.7\textwidth]{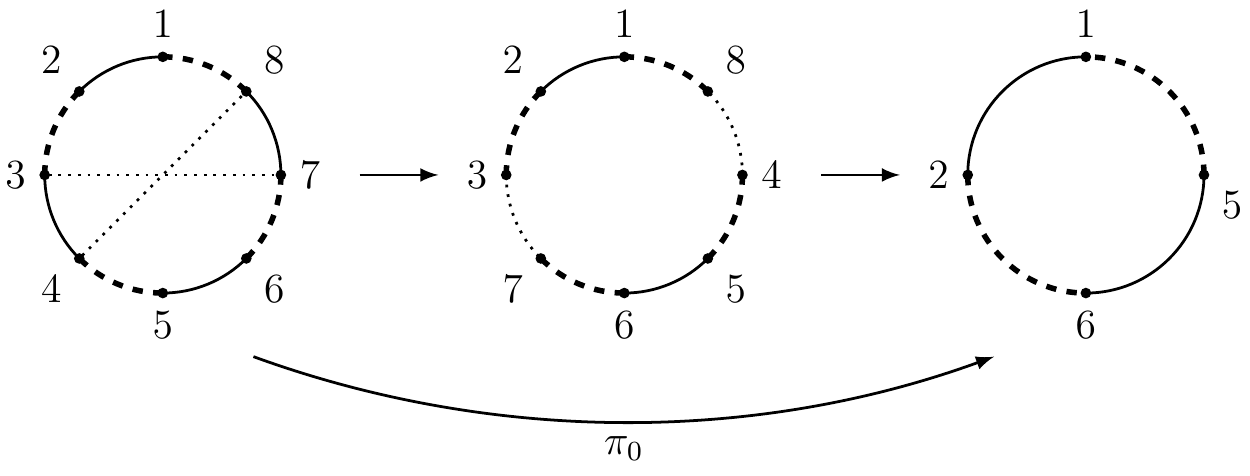}
			\caption{Visualisation of the projection map $\pi_0$ for the case of diagram $D_3$ from \reffig{img_chordexamples}. } 
			\label{img_proj}
		\end{center}
	\end{figure}
	
	The final notion that we need to introduce is the signum $\sgn(u,v)$ of two free vertices $u,v \in V_D^{(2)}$ within a chord diagram. It is $-1$ if they are not part of the same tricoloured cycle of $D$, and $0$ or $1$ if there are an even or odd number of paths between them. Alternatively, in terms of $\pi_0(D)$, $\sgn(u,v)=-1$ if $u$ and $v$ are in different base cycles and is $0$ or $1$ if they are in the same base cycle and are separated by an even or odd number of base edges. In \cite[Prop. 3.5]{Golz_2017_Traces} it was worked out how the numbers $c_2$ and $c_3$ change when a chord is added to a chord diagram $D_0\in D^n_k$ with $k<N$. Focussing only on the total number $\tilde c$ it reduces to
	\al{
		\tilde c(D) = \tilde c(D_0) + \sgn(u,v),
	}
	where $D = ( V_{D_0}, E_{D_0}^0 \cup \{u,v\}, E_{D_0}^1, E_{D_0}^2)$.
	
	%  Feynman graphs
	\subsubsection{Chord diagrams and Feynman graphs} \label{sec_cd_fg}
		
		In this section we establish the connection between Feynman graphs (and integrals) and chord diagrams. For concreteness we focus on the case of photon propagator  graphs with a single fermion cycle, in Feynman gauge, and we ignore subdivergences.\\
	
		The Feynman rules for a fermion cycle yield a trace
		\al{
			\tr( \gamma_{\mu_1}\dotsm\gamma_{\mu_{4h_1}} )
		}
		where the matrices $\gamma_{\mu_i}$ correspond to fermion edges ($i$ odd) and vertices ($i$ even), and $h_1\equiv h_1(\Gamma)$ is the graph's first Betti number. Since we are in Feynman gauge every other matrix is contracted with metric tensors $g_{\mu_i\mu_j}$, corresponding to the photon propagators (including the external photon, see sec. \ref{subsec_contract}). The trace can be visualised as a chord diagram $D_{\Gamma} \in \cD^{2h_1}_{h_1}$ in which each vertex is labelled by one matrix and contraction via metric tensors is represented by chords.
		
		We are now interested in sums over chord diagrams that result from all possible additions of further chords to $D_{\Gamma}$. Because the chords fixed in place are always the same, we can consider smaller diagrams instead, namely diagrams built on the projection $D_{\Gamma}^0=\pi_0(D_{\Gamma})$. Even if $D_{\Gamma}$ has only one base cycle the projection may contain multiple base cycles. Hence, let $n\in\NN_+^{\ell_0}$ with $\sum n_i = N = h_1$ be some suitable tuple representing the base cycle structure of $D_{\Gamma}^0$ after the projection, such that $D_{\Gamma}^0 \in \cD^n_0$. Then we denote with $\cD^k_{\Gamma} \simeq \cD^n_{h_1-k}$ the set of all chord diagrams that contain the base cycles of $D_{\Gamma}^0\in \cD^n_0$ as a subgraph together with $h_1-k$ chords and have their vertices labelled by fermion edges of the underlying Feynman graph. Each such diagram corresponds to a trace of $4h_1$ Dirac matrices contracted with $2h_1-k$ metric tensors\footnote{To be precise, these diagrams of course correspond to some product of traces with a total of $2h_1$ Dirac matrices. However, the fixed chords in $D_{\Gamma}$ do not influence $\tilde c$, so we can just take the factor $(-2)^{h_1}$ due to these $h_1$ chords and otherwise work with the smaller diagrams, even though we actually compute the contraction of a larger product of matrices. See \cite{Golz_2017_Traces} for details.},
		\al{
			\cD^k_{\Gamma} \ni D \sim
							\bigg(\prod_{ (u,v)\in E_D^0 } g_{\mu_u\mu_v} \bigg)
							\bigg(\prod_{ (u',v')\in E_{D_{\Gamma}}^0 } g_{\mu_{u'}\mu_{v'}} \bigg) 
							\tr( \gamma^{\mu_1}\dotsm\gamma^{\mu_{4h_1}} ).
		}
		For $k=0$ this trace is simply $(-2)^{1+2h_1+\tilde c(D)}$ \cite[Theorem 3.9]{Golz_2017_Traces}. Similar results also hold if there are uncontracted matrices left, and for contractions between products of multiple traces. Since there is only a single external momentum $q$ all matrices not contracted with metric tensors are contracted to $\slashed q = q_{\rho}\gamma^{\rho}$, and with $\slashed q\slashed q = q^2$ one finds that one always just has an integer multiple of a power of $q^2$. Applying this to the parametric integrand for QED Feynman integrals from  \cite{Golz_2017_CyclePol}, which we will do in more detail in section \ref{sec_application}, yields sums of the form
		\al{
			I_{\Gamma}^{(k)} \propto \sum_{D\in \cD_{\Gamma}^k} (-2)^{\tilde c(D)} \bigg(\prod_{ (u,v)\in E_D^0 } \cp{u}{v}\bigg) \prod_{ w\in V_D^{(2)} } x_{\Gamma}^w, \label{eq_sumdetail}
		}
		which we will be able to rewrite with the two main theorems of this article.

%2
% DOGSON POLYNOMIALS
% DP revisited
\section{Dodgson polynomials revisited}

	In order to prove the polynomial identities of section \ref{sec_Pol_id} we will need a variant of the Dodgson polynomials that is in some sense a reinterpretation (in section \ref{sec_dp_01}) but also a generalisation (in section \ref{sec_dp_02}). This then allows us to define what we call partition polynomials in section \ref{sec_defpp}.

	% dodgson cycle polynomials
	\subsection{Dodgson cycle polynomials}\label{sec_dp_01}
	
	The relation $\cp[G]{e}{e} = \kp[G\setminus e] = \kp[G]^{\{e\},\{e\}}$ suggests a possible connection between cycle and Dodgson polynomials, and indeed we find
	\begin{proposition}\label{05a_prop_eqpol}
		\al{
			\cp[G]{i}{j} = \pm\kp[G]^{\{i\},\{j\}}
		}
		for all $i,j\in E_G$.
	\end{proposition}
	\begin{proof}
		For $i=j$ the proof is done and for $i\neq j$ we use the combinatoric interpretation from \refeq{eq_DodgsonCombInt},
		\al{
			\kp[G]^{\{i\},\{j\}} = \sum_{T\subset E_G\setminus\{i,j\}} \pm \prod_{e\notin T} \alpha_e.
		}
		A sum over spanning trees can be decomposed into a double sum over paths $P\subset G\setminus i$ and spanning trees of the corresponding graph $\con{(G\setminus i)}{P}$ where all paths are between endpoints $\partial_+(i)$ and $\partial_-(i)$ and contain the edge $j$. Then adding $i$ to each path completes it into a simple cycle $C_P = P\cup\{i\} \in \cC_G^{[1]}$ that contains both $i$ and $j$, and the corresponding monomials of $\cp[G]{i}{j}$ and $\kp[G]^{\{i\},\{j\}}$ indeed agree, at least up to sign. The signs $\ro_{C_1}(i,j)$, $\ro_{C_2}(i,j)$ of two partial polynomials $\kp[\con{G}{C_1}]$ and $\kp[\con{G}{C_2}]$ in $\cp[G]{i}{j}$ differ if and only if $C_1\cup C_2$ is -- up to contraction of longer paths to single edges -- isomorphic to $K_4$: 
	 \begin{center}
	 	\includegraphics[width=0.8\textwidth]{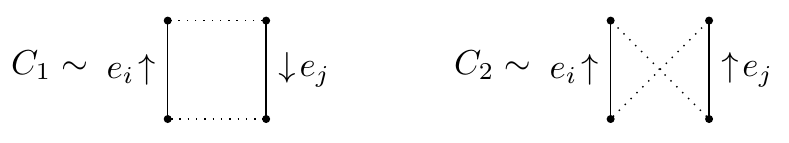}
	 \end{center}
	 Comparing this with the discussion of signs in Dodgson polynomials in section 2 of \cite{BrownYeats_2011_SpanningForest}, one finds that the endpoints of $i$ are precisely the transposed vertices given in \cite[corollary 17]{BrownYeats_2011_SpanningForest} as a criterion for opposite signs. Therefore all partial polynomials have the correct relative signs and only the overall sign ambiguity of Dodgson polynomials remains, concluding the proof.
	\end{proof}

	It should be noted that the sign ambiguity of the Dodgson polynomials is of course not entirely absent from the cycle polynomials -- the choice one has to make is simply moved from the order of rows and columns in a matrix to the orientations of edges in $G$. Since we always considered our graphs together with some such fixed choice from the very beginning it does not appear in the combinatorial definition of the cycle polynomials. Moreover, in the context of Feynman integrals we can even have a physical motivation for certain orientations, e.g. aligning all fermion edge orientations with fermion flow. 
	
	We can use this to fix the choice of the graph matrix such that the signs of $\cp[G]{i}{j}$ and $\kp[G]^{\{i\},\{j\}}$ agree. Furthermore, the interpretation of cycle polynomials as a fixed-sign version of Dodgson polynomials also suggests the definition of a higher order cycle polynomial via the Dodgson identity \refeq{eq_DodgsonID}.

	% Dodgson Cycle polynomials
	\begin{definition}\label{02c_def_dodgsoncyclepol}
		Let $G$ be a connected graph and $\cp[G]{i}{j}$ for all $i, j\in E_G$ the cycle polynomial as defined in \refeq{eq_defcycpol}. Then define an alphabet $\sA = \{ \sa_i \ | \ i \in E_G \}$ in which each letter is associated to an edge of $G$ and consider two words $\su,\sv$ over this alphabet with $|\su|=k=|\sv|$. The Dodgson cycle polynomial is then defined as $\cp[G]{\sa_i}{\sa_j} \defeq \cp[G]{i}{j}$ if $k=1$ and
		\al{
			\cp[G]{\su}{\sv} \defeq \kp[G]^{1-k} \sum_{\sigma \in S_k} \sgn(\sigma) \prod_{i=1}^{k} \cp[G]{\su_i}{\sigma_i(\sv)}, \label{eq_DefExpansion}
		}
		where $\su_i$, $\sigma_i(\sv)$ denote the $i$-th letter of $\su$ and (the permutation of) $\sv$, for $2\leq k\leq h_1(G)$.
	\end{definition}
	
	In this we simply recursively define $\cp[G]{\su}{\sv}$ for words of length $k$ by repeatedly using the Dodgson identity \refeq{eq_DodgsonID} or its generalisations derived from \refeq{eq_DJ}. For $k=2,3$ one has
	\al{
		\cp[G]{\sa_1\sa_2}{\sa_3\sa_4} &= \kp[G]^{-1} \big( \cp[G]{\sa_1}{\sa_3}\cp[G]{\sa_2}{\sa_4} - \cp[G]{\sa_1}{\sa_4}\cp[G]{\sa_2}{\sa_3}\big),\\[5mm]\nonumber
		\cp[G]{\sa_1\sa_2\sa_3}{\sa_4\sa_5\sa_6} &= \kp[G]^{-2} \big( \cp[G]{\sa_1}{\sa_4}\cp[G]{\sa_2}{\sa_5}\cp[G]{\sa_3}{\sa_6}
														- \cp[G]{\sa_1}{\sa_4}\cp[G]{\sa_2}{\sa_6}\cp[G]{\sa_3}{\sa_5}\\\nonumber
										&\qquad\qquad		+\cp[G]{\sa_1}{\sa_5}\cp[G]{\sa_2}{\sa_4}\cp[G]{\sa_3}{\sa_6}
														- \cp[G]{\sa_1}{\sa_5}\cp[G]{\sa_2}{\sa_6}\cp[G]{\sa_3}{\sa_4}\\
										&\qquad\qquad		+\cp[G]{\sa_1}{\sa_6}\cp[G]{\sa_2}{\sa_4}\cp[G]{\sa_3}{\sa_5}
														- \cp[G]{\sa_1}{\sa_6}\cp[G]{\sa_2}{\sa_5}\cp[G]{\sa_3}{\sa_4} \big).
	}
	Note that this also permits an expansion (essentially the cofactor expansion of the determinant) that yields, e.g. for $k=3$	,
	\al{		
		\cp[G]{\sa_1\sa_2\sa_3}{\sa_4\sa_5\sa_6}	 = \kp[G]^{-1} \Big( \cp[G]{\sa_1}{\sa_4}\cp[G]{\sa_2\sa_3}{\sa_5\sa_6}
														- \cp[G]{\sa_1}{\sa_5}\cp[G]{\sa_2\sa_3}{\sa_4\sa_6}
														+ \cp[G]{\sa_1}{\sa_6}\cp[G]{\sa_2\sa_3}{\sa_4\sa_5} \Big).
	}
	which will be very useful later on.\\

	Defining the polynomials like this imposes an ordering on the indices instead of using unordered sets $I$, $J$. This yields a symmetry $\cp[G]{\su}{\sv} = \sgn(\sigma)\cp[G]{\su}{\sigma(\sv)}$ for all permutations of letters in the words, which we will be able to exploit for our purposes below. Moreover, note that $\cp[G]{\su}{\sv} = 0$ if one of the words contains a repeated letter and
	\al{
		\cp[G]{\su}{\sv} = (-1)^{l+m}\cp[G\setminus e]{\su_1\dotsm\hat\su_l\dotsm\su_k}{\sv_1\dotsm\hat\sv_m\dotsm\sv_k}
	}
	if $\su_l = \sv_m = \sa_e$.

	\begin{remark}
		The relation between Dodgson polynomials and \textit{``sums over subgraphs of $G$ containing cycles which satisfy certain properties''} was already observed by Brown when he originally defined Dodgson polynomials in \cite[Remark 24]{Brown_2010_Periods}, but not further pursued. At some point it might be worthwhile to study what exactly these certain properties should be for higher order polynomials, so one can give a direct combinatorial definition analogous to \refeq{eq_defcycpol}, but for this article the recursive definition above shall suffice.	
	\end{remark}
	
	% vertex indexed 
	\subsection{Vertex-indexed Dodgson polynomials}\label{sec_dp_02}
	
	We just modified the Dodgson polynomial in a way that allows us to control their signs by relating them to another polynomial also indexed by edge subsets of the underlying graph. However, the graph matrix is an $E_G+V_G-1$ square matrix that also has rows corresponding to vertices of the graph. It is therefore quite natural to extend the definition of the Dodgson polynomials to include deletion of rows and columns labelled by vertices. Since the determinant identity \refeq{eq_DJ} holds generally, irrespective of which columns or rows are deleted, the polynomials given by such minors still satisfy the Dodgson identity. Moreover, we have already seen the fixed-sign versions of these types of Dodgson polynomials that we can use analogously to the cycle polynomials in the previous section. Remember the spanning forest polynomials used to rewrite the second Symanzik polynomial in \refeq{eq_ssp_stmat}. Using again the block matrix identity from \refeq{eq_blockID} we can write
	\al{
		\Phi_G^{\{v_0\},\{v_1,v_2\}} = (-1)^{v_1+v_2} \alpha_{E_{\Gamma}} \det( \tilde {L'}^{\{v_1\}}_{\{v_2\}} ) = (-1)^{v_1+v_2}\det( M(G)^{\{v_1\}}_{\{v_2\}} ).
	}
	The ambiguous sign of the determinant is precisely cancelled by the factor $(-1)^{v_1+v_2}$ such that $\Phi_G^{\{v_0\},\{v_1,v_2\}}$ is indeed a fixed-sign version of the Dodgson polynomial $\kp[G]^{\{v_1\},\{v_2\}}$. Hence, we reuse our previous notation to define
	\al{
		\cp[G]{\sa_{v_1}}{\sa_{v_2}} \defeq \Phi_G^{\{v_0\},\{v_1,v_2\}},
	}
	now with words (over an extended alphabet that includes vertices) indexing it, as in the previous case of cycle Dodgson polynomials. In this notation the Dodgson identity again takes the form
	\al{
		\cp[G]{\sa_1}{\sa_2} \cp[G]{\sa_3}{\sa_4} - \cp[G]{\sa_1}{\sa_3}\cp[G]{\sa_2}{\sa_4} = \kp \cp[G]{\sa_1\sa_4}{\sa_2\sa_3},	\label{eq_DI_rewrite}
	}
	and generalisations are analogous to \refeq{eq_DefExpansion}. Note that, where edge indices lower the degree of the polynomial, such that $\deg( \cp[G]{\sw_1}{\sw_2} ) = h_1(G)-|\sw_i|$, if the letters of both $\sw_i$ correspond to edges, the vertex indices do the opposite. For single letters, $\deg(\cp[G]{\sa_i}{\sa_j}) = \deg( \Phi_G^{\{v_0\},\{i,j\}} ) = h_1(G)+1$, such that the polynomial with two-letter words on the r.h.s. has to have degree $h_1+2$. Another property we get by courtesy of the spanning forest polynomial is that
	\al{
		\cp[G]{\sa_{v}}{\sa_{v}} = \kp[G|_{v=v_0}].
	}
	In other words, equal indices correspond to identification of that vertex with $v_0$ in the graph, analogous to the edge-indexed case $\cp[G]{\sa_{e}}{\sa_{e}} = \kp[G\setminus e]$, which indicated deletion of an edge.\\
	
	In $\Phi_G^{\{v_0\},\{v_1,v_2\}}$ the vertex $v_0$ whose row and column are initially deleted from the graph matrix is explicit. We will see below that it is actually useful to consider Dodgson polynomials coming from different such choices. Hence, from now we will use the subscript $\cp[G,v_0]{\su}{\sv}$ to indicate it, whenever the choice actually matters. Note that this is different from the subscript $K$ in the usual Dodgson polynomial $\kp[G,K]^{I,J}$, which indicates contracted edges and is always empty for us.

	% Partition polynomials
	\subsection{Partition polynomials}\label{sec_defpp}

	With these new variants of the Dodgson polynomials we can define another new polynomial that bridges the gap between Feynman graphs and the chord diagrams associated to them. For this purpose, we briefly return to the case of Dodgson polynomials only indexed by edge subsets of the underlying graph, not vertices.\\

	Let $\Gamma$ be a suitable Feynman graph such that $D_{\Gamma}^0 \equiv D\in \cD^n_0$ with $n\in \NN^{\ell}$ and label all its vertices with the graph's fermion edges, i.e. letters from an alphabet $\sA=\{ \sa_i \ | i \in E_{\Gamma}^{(f)} \}$. Consider all pairs of monomial words $(\su, \sv)$ of length $|\su| = N = |\sv|$ over this alphabet such that $\su\sv$ contains each letter exactly once. Then the symmetries
	\al{
		\cp{\su}{\sv} = \cp{\sv}{\su} \quad\text{ and }\quad \cp{\su}{\sv} = \sgn(\sigma) \cp{\su}{\sigma(\sv)}  \qquad  \forall \sigma\in S_N,
	}
	induce an equivalence relation on these words via
	\al{
		(\su,\sv) \sim (\su',\sv') \iff \cp{\su}{\sv} = \pm\cp{\su'}{\sv'},
	}
	or equivalently
	\al{
		(\su,\sv) \sim (\su',\sv') \iff \exists\ \sigma, \sigma' \in S_N \text{ s.t. } \su'=\sigma(\su),\ \sv'=\sigma'(\sv).
	}
	Let $\sP$ denote the corresponding set of equivalence classes of pairs $(\su,\sv)$ that satisfy the above mentioned properties. For the two coloured subsets of base edges $E_D^1$ and $E_D^2$ define the corresponding subsets $\sP_i\subset \sP$ by imposing an additional constraint: For all edges $(u,v) \in E_D^i$ we demand that the two corresponding letters do not appear in the same word, i.e. $\sa_u \in \su$ and $\sa_v \in \sv$ or vice versa. The full set of equivalence classes is then the union $\sP = \sP_1 \cup \sP_2$. Moreover, in most cases the $\sP_i$ intersect only in exactly one element, which, assuming the vertices of $D$ are labelled consecutively within each base cycle, is the class of pairs that contain all letters labelled with odd numbers in one word and those labelled with even numbers in the other. The only exception occurs if $D$ has one or more base cycles of size $1$. Then there is a base edge of either colour between the same two vertices, leading to some redundancy. In particular, $\sP_1 = \sP_2$ if $n=(1,\dotsc,1)$. 
	
	Finally, we need to fix one distinguished representative of each class with respect to which we consider permutations. Assuming some arbitrary ordering of $i$-coloured base edges $(u_1,v_1), \dotsc$, $(u_N,v_N) \in E_D^i$ each equivalence class contains exactly one element that we notate $(\su_{\id},\sv_{\id})$ such that $\sa_{u_j}$ and $\sa_{v_j}$ are the $j$-th letters of $\su_{\id}$ and $\sv_{\id}$, or vice versa. For any other ordering of base edges the designated element would be related to $(\su_{\id}, \sv_{\id})$ by the same permutation in both words, such that the choice of ordering on $E_D^i$ does not matter.\\

	For all $(\su,\sv)\in \sP$ and partitions of $i$-coloured base edges $\cE = (E_1,\dotsc,E_{|\cE|}) \in \cP(E_D^i)$ define a map $\lambda_{\cE}$ as follows. Let
	\al{
		V_j \defeq \bigcup_{(u,v) \in E_j} \{u,v\} \subseteq V_D \label{eq_PartVert}
	}
	be the set of vertices in the part $E_j$ and consider the restriction
	\al{
		(\su_j,\sv_j) = (\su, \sv)|_{ \sa_k = 1 \ \forall k \in V_D\setminus V_j}
	} 
	of $(\su,\sv)$ to the alphabet corresponding to these vertices. In each $(\su_j,\sv_j)$ all letters not associated to this part of the partition are removed but, critically, the order of the remaining letters is preserved. Then
	\al{
		\lambda_{\cE}(\su,\sv) \defeq \begin{cases}
							\{ (\su_1,\sv_1), \dotsc, (\su_{|\cE|},\sv_{|\cE|}) \}	& \text{ if } |\su_j| = |\sv_j| \text{ for all } 1\leq j \leq |\cE|,\\
							\hspace{5em} \emptyset		& \text{ else.}
						\end{cases}	\label{eq_deflambda}
	}
	The concatenations $\su_1\dotsm\su_{|\cE|}$ and $\sv_1\dotsm\sv_{|\cE|}$ are then permutations of $\su$ and $\sv$ (which are themselves permutations of the words $\su_{\id}, \sv_{\id}$ of their equivalence class) and we define
	\al{
		\sgn_{\cE}(\su,\sv) \defeq \begin{cases} \hspace{2.5em} 0 & \text{ if } \lambda_{\cE}(\su,\sv) = \emptyset, \\
									\sgn(\sigma)\sgn(\sigma') & \text{ else,}
							\end{cases}
							\label{eq_defsgn}
	}
	where $\sigma, \sigma' \in S_N$ are the permutations with $\sigma(\su_{\id}) = \su_1\dotsm\su_{|\cE|}$ and $\sigma'(\sv_{\id}) = \sv_1\dotsm\sv_{|\cE|}$. With this we are now ready to insert these types of words into certain combinations of Dodgson polynomials, which we will call partition polynomials.

% DEF partition polynomials
\begin{definition}\label{def_PartPol}
	Let $\Gamma$ be a QED Feynman graph with the associated chord diagram $D_{\Gamma}$ such that $D\equiv \pi_0(D_{\Gamma}) = D_{\Gamma}^0 \in \cD^n_0$ with $n\in \NN^{\ell}$, $N=\sum_i n_i = h_1(\Gamma)$. Then we define the partition polynomial of $\Gamma$ to be
	\al{
		Z_{\Gamma}^0(\alpha) &\defeq \sum_{\cE \in \cP(E_D^1)} (-\kp)^{N-|\cE|} (|\cE|+1)! \sum_{ (\su,\sv)\in \sP_2 } \sgn_{\cE}(\su,\sv) \nquad
		\prod_{ (\su',\sv')\in\lambda_{\cE}(\su,\sv)}\nqquad\cp{\su'}{\sv'},
		\label{eq_PartPol}
	}
	where $\cP(E_D^1)$ is the set of all partitions of $1$-coloured base edges of $D$. Moreover, for $1\leq l \leq N$ let
	\al{
		Z_{\Gamma}^0\big|_l \defeq \sum_{\substack{\cE \in \cP(E_D^1)\\ |\cE| = l}} \sum_{ (\su,\sv)\in \sP_2 }\sgn_{\cE}(\su,\sv) \nquad \prod_{ (\su',\sv')\in\lambda_{\cE}(\su,\sv) }\nqquad\cp{\su'}{\sv'}
	}
	such that
	\al{
		Z_{\Gamma}^0 = \sum_{l=1}^N (-\kp)^{N-l} (l+1)!\ Z_{\Gamma}^0\big|_l.\label{eq_ppshort}
	}
\end{definition}
	Note that using partitions $\cP(E_D^2)$ in the first and words $(\su,\sv)\in \sP_1$ in the second sum yields the exact same polynomial. This symmetry is not quite obvious from this definition but will become so in the proof of theorem \ref{theo_main} below.	However, the separate polynomials $Z_{\Gamma}^0\big|_l$ do differ considerably depending on whether one sums over $\cP(E_D^1)$ and $\sP_2$ or $\cP(E_D^2)$ and $\sP_1$. Hence, when discussing these polynomials specifically, one should make clear which one is chosen. We reiterate that the sum $Z_{\Gamma}^0$ is independent of this choice, which only reflects two different possible decompositions.\\
	
	Based on this definition we can introduce a similar polynomial that incorporates vertex-indexing in Dodgson polynomials. For the purposes of this article it suffices to stick to a very specific vertex indexing, but it should certainly be possible to extend this to include any type of Dodgson polynomial.

	% DEF partition polynomials_1
\begin{definition}\label{def_PartPol_1}
	Let everything be as in def. \ref{def_PartPol}. Additionally, let $\Gamma$ be a Feynman graph with only two non-zero external momenta and $x,y \in V_{\Gamma}^{ext}$ the corresponding external vertices. Let $\sy$ be the additional letter representing $y$ and assume that the deleted column and row of the graph matrix corresponds to $x$, i.e. all Dodgson polynomials are $\cp{\su}{\sv} \equiv \cp[\Gamma,x]{\su}{\sv}$. Define
	\al{\nonumber
		Z_{\Gamma}^1\big|_l &\defeq 
			\sum_{\substack{\cE \in \cP(E_D^1)\\ |\cE| = l}} \sum_{ (\su,\sv)\in \sP_2 } \sgn_{\cE}(\su,\sv) \nquad
			\prod_{ (\su',\sv')\in\lambda_{\cE}(\su,\sv)}\nqquad\cp{\su'}{\sv'}\nquad \sum_{ (\su',\sv')\in\lambda_{\cE}(\su,\sv) } 
			\bigg( \frac{\cp{\su'\sy}{\sv'\sy}}{\cp{\su'}{\sv'}} - \frac{\vssp}{\kp} \bigg)\\[3mm]
			&= -l\frac{\vssp}{\kp}Z_{\Gamma}^0\big|_l + \sum_{\substack{\cE \in \cP(E_D^1)\\ |\cE| = l}} \sum_{ (\su,\sv)\in \sP_2 } \sgn_{\cE}(\su,\sv) \nquad
			\prod_{ (\su',\sv')\in\lambda_{\cE}(\su,\sv)}\nqquad\cp{\su'}{\sv'}\nquad \sum_{ (\su',\sv')\in\lambda_{\cE}(\su,\sv) } \frac{\cp{\su'\sy}{\sv'\sy}}{\cp{\su'}{\sv'}}.
			\label{eq_ppol1_aux}
	}
	Then we define the first order partition polynomial of $\Gamma$ to be
	\al{
		Z_{\Gamma}^1(\alpha) &\defeq \frac{1}{2} \sum_{l=1}^N (-\kp)^{N-l+1} (l+1)!\ Z_{\Gamma}^1\big|_l.
			\label{eq_PartPol_1}
	}
\end{definition}
	
	Note the additional factors of $1/2$ and $-\kp$, in contrast to \refeq{eq_PartPol} above. Together with the observation that $\vssp=\cp{\sy}{\sy}$ and $\kp = \cp{\emptyset}{\emptyset}$ in the first line of \refeq{eq_ppol1_aux} this suggests a straightforward generalisation
	\al{
		Z_{\Gamma}^k(\alpha) &\defeq \frac{1}{2^k} \sum_{l=1}^N (-\kp)^{N-l+k} (l+1)!\ Z_{\Gamma}^k\big|_l.
	}
	$Z_{\Gamma}^k\big|_l$ should contain something like a sum over all choices of $k$ word pairs in $\lambda_{\cE}(\su,\sv)$ to which the letter $\sy$ is added. Then the factor $1$ in $Z_{\Gamma}^0\big|_l$ corresponds to a sum over the unique choice of no element at all and the sum in $Z_{\Gamma}^1\big|_l$ is the sum over choices of exactly one word pair. If this is in fact a correct (i.e. useful) generalised definition shall be studied in future work. For now we will concentrate on the cases of order $0$ and $1$.

% 3 
% SUMMATION

% SUMMATION
\section{Polynomial identities}\label{sec_Pol_id}

	The statement of our two main theorems is now that the two partition polynomials $Z_{\Gamma}^0$ and $Z_{\Gamma}^1$ are in fact equal to the sums of chord diagrams, with products of cycle polynomials in each summand, that appear in the parametric integrand of QED.

%
%
%
% FIRST SUMMATION
\subsection{The first summation theorem}
	
\begin{theorem}\label{theo_main}
	%Let $D_0\in \cD^n_0$ with $n\in \NN^{\ell}$, $N=\sum n_i$. Then ^n_N
	\al{
		Z_{\Gamma}^0 = \frac{1}{2} \sum_{D\in \cD_{\Gamma}^0} (-2)^{\tilde c(D)} \nquad\prod_{(u,v)\in E_D^0} \cp{\sa_u}{\sa_v}.
		 \label{eq_maintheo}
	}
\end{theorem}

	\noindent In order to prove this we first need some auxiliary results. First we attempt to study the summation by essentially working backwards and looking at sums $\sum \cp{\su_{\id}}{\sv_{\id}}$ for $(\su_{\id},\sv_{\id})\in \sP_2$, which appear in the partition polynomial for the single part partition $\cE = \{E_D^1\}$.

% Lemma1
\begin{lemma}\label{lemma_fullpart}
	Let $\sP_j$ be as above and $c_2^j(D) \defeq |\cC_D^{0j}|$ the number of two-coloured cycles consisting of chords and $j$-coloured base edges (such that $c_2(D) = c_2^1(D)+c_2^2(D)$). Then
	\al{
		\sum_{ (\su_{\id},\sv_{\id})\in \sP_j } \cp{\su_{\id}}{\sv_{\id}} 
			= (-\kp)^{1-N}\sum_{D\in\cD_{\Gamma}^0} (-2)^{c_2^j(D)-1} \nquad\prod_{(u,v)\in E_D^0} \cp{\sa_u}{\sa_v},\label{eq_lemma_fullpart}
	}
\end{lemma}
\begin{proof}
	Quick computations show that the claim holds for all $n$ with $N=\sum n_i = 1, 2$, and even $N=3$ is only mildly tedious, as shown below in example \ref{ex_lem1}. We now reduce the l.h.s. of \refeq{eq_lemma_fullpart} to a sum over expressions corresponding to $N-1$, in order to prove by induction.
	
	Consider a word pair $( \sx_{11}\dotsm\sx_{1N}, \sx_{21}\dotsm\sx_{2N} )$ with all $\sx_{ij}\in \sA$. Assuming this word is a representative $(\su_{\id},\sv_{\id})\in \sP_j$, each pair $(\sx_{1k},\sx_{2k})$ of $k$-th letters corresponds to a base edge of $E_{D_0}^j$, for a chord diagram $D_0\in \cD_0^n$. With \refeq{eq_DefExpansion} its Dodgson polynomial can be written as
	\al{\nonumber
		\kp&\cp{\sx_{11}\dotsm\sx_{1N}}{\sx_{21}\dotsm\sx_{2N}}\\\nonumber
				&\quad= \sum_{k=1}^{N} (-1)^{1+k} \cp{\sx_{11}}{\sx_{2k}}
				\cp{\sx_{12}\dotsm\sx_{1N}}{\sx_{21}\dotsm\hat\sx_{2k}\dotsm\sx_{2N}}\\[3mm]
				&\quad= \cp{\sx_{11}}{\sx_{21}} \cp{\sx_{12}\dotsm\sx_{1N}}{\sx_{22}\dotsm\sx_{2N}}
				-\sum_{k=2}^{N} \cp{\sx_{11}}{\sx_{2k}}
				\cp{\sx_{1k}\sx_{12}\dotsm\hat\sx_{1k}\dotsm\sx_{1N}}{\sx_{21}\dotsm\hat\sx_{2k}\dotsm\sx_{2N}}.\label{eq_Dodgson_exp}
	}
	Moving the letter $\sx_{1k}$ in the last line guarantees that the letter pairs $(\sx_{1l},\sx_{2l})$, with $l\neq 1,k$, are still paired up in the expansion. In fact, the word pairs
	\al{
		(\sx_{1k}\sx_{12}\dotsm\hat\sx_{1k}\dotsm\sx_{1N}, \sx_{21}\dotsm\hat\sx_{2k}\dotsm\sx_{2N})
	}
	are the representatives $(\su_{\id}', \sv_{\id}')$ of an equivalence class of word pairs associated to the diagram $\pi_0(D)$, where $D$ is $D_0$ together with the chord corresponding to the letter pair $(\sx_{11},\sx_{2k})$. The sum over all equivalence classes in $\sP_j$ can be realised by summing word pairs of the form
	\al{
		( \sx_{(1+t_1)1}\dotsm\sx_{(1+t_N)N},\ \sx_{(2-t_1)1}\dotsm\sx_{(2-t_N)N} )
	}
	over all $N$-tuples in $\cT=\{ t\in \{0,1\}^N \ | \ t_1=0\}$. One finds
	\al{\nonumber
		\kp&\sum_{t\in\cT} \cp{\sx_{(1+t_1)1}\dotsm\sx_{(1+t_N)N}}{\sx_{(2-t_1)1}\dotsm\sx_{(2-t_N)N}}\\\nonumber
				&\quad= \cp{\sx_{11}}{\sx_{21}} \sum_{t\in\cT} \cp{\sx_{(1+t_2)2}\dotsm\sx_{(1+t_N)N}}{\sx_{(2-t_2)2}\dotsm\sx_{(2-t_N)N}}\\
				&\qquad -\sum_{k=2}^{N}\sum_{t\in\cT} \cp{\sx_{11}}{\sx_{(2-t_k)k}}
		\cp{\sx_{(1+t_k)k}\sx_{(1+t_2)2}\dotsm\hat\sx_{(1+t_k)k}\dotsm\sx_{(1+t_N)N}}{\sx_{21}\dotsm\hat\sx_{(2-t_k)k}\dotsm\sx_{(2-t_N)N}}.
		\label{eq_lemma_expansion}
	}
	Now we want to translate this back to vertices of a chord diagram. Let $u,v\in V_{D_0}$ such that $\sx_{11}=\sa_u$, $\sx_{21}=\sa_v$ and $(u,v) \in E_{D_0}^j$. Note that, by definition of $\sP_j$, such $u,v$ always exist. Then \refeq{eq_lemma_expansion} becomes
	\al{
		\kp\nquad\sum_{ (\su_{\id},\sv_{\id})\in \sP_j } \cp{\su_{\id}}{\sv_{\id}}
			&= 2\cp{\sa_u}{\sa_v} \nquad\sum_{ (\su_{\id}', \sv_{\id}') \in \sP_j^{u,v}}\nquad\cp{\su_{\id}'}{\sv_{\id}'} 
				-\sum_{\substack{ w\in V_{D_0}\\ w\neq u,v }} \cp{\sa_u}{\sa_w} 
					\nquad\sum_{ (\su_{\id}', \sv_{\id}') \in \sP_j^{u,w}}\nquad\cp{\su_{\id}'}{\sv_{\id}'},\label{eq_prop_exp}
	}
	where $\sP_j^{u,v}$ and $\sP_j^{u,w}$ are the classes of word pairs after addition of the chords $(u,v)$ or $(u,w)$ respectively. Replacing these sums with the corresponding r.h.s. of \refeq{eq_lemma_fullpart} finishes the proof, where the factor of $-(-2)$ in the first term corresponds to the addition of the cycle that consists of the $j$-coloured base edge $(u,v)$ and the chord between those same vertices. All other chords $(u,w)$ added to $D_0$ do not add two-coloured cycles but only split, twist or merge base cycles when projected out with $\pi_0$.
\end{proof}

% example
\begin{example}\label{ex_lem1}
	Consider as an example $N=3$ with a single base cycle. Label vertices consecutively from $1$ to $6$ and choose $j$ to be the colour of $(1,2)$. Then the sum over word pairs in $\sP_j$ on the l.h.s. of \refeq{eq_lemma_fullpart} is
	\al{
		\cp{\sa_1\sa_3\sa_5}{\sa_2\sa_4\sa_6} + \cp{\sa_1\sa_4\sa_5}{\sa_2\sa_3\sa_6} + \cp{\sa_1\sa_3\sa_6}{\sa_2\sa_4\sa_5} + \cp{\sa_1\sa_4\sa_6}{\sa_2\sa_3\sa_5}.
	}
	Expanding each term as defined in \refeq{eq_DefExpansion} yields $\kp^{-2}$ times 24 terms, 15 of which are distinct, such that one finds
	\al[*]{
		4&\cp{\sa_1}{\sa_2}\cp{\sa_3}{\sa_4}\cp{\sa_5}{\sa_6} - 2\cp{\sa_1}{\sa_2}\cp{\sa_3}{\sa_5}\cp{\sa_4}{\sa_6} - 2\cp{\sa_1}{\sa_2}\cp{\sa_3}{\sa_6}\cp{\sa_4}{\sa_5}\\[2mm]
		-2&\cp{\sa_1}{\sa_3}\cp{\sa_2}{\sa_4}\cp{\sa_5}{\sa_6} + \cp{\sa_1}{\sa_3}\cp{\sa_2}{\sa_5}\cp{\sa_4}{\sa_6} + \cp{\sa_1}{\sa_3}\cp{\sa_2}{\sa_6}\cp{\sa_4}{\sa_5}\\[2mm]
		-2&\cp{\sa_1}{\sa_4}\cp{\sa_2}{\sa_3}\cp{\sa_5}{\sa_6} + \cp{\sa_1}{\sa_4}\cp{\sa_2}{\sa_5}\cp{\sa_3}{\sa_6} + \cp{\sa_1}{\sa_4}\cp{\sa_2}{\sa_6}\cp{\sa_3}{\sa_5}\\[2mm]
		 +&\cp{\sa_1}{\sa_5}\cp{\sa_2}{\sa_3}\cp{\sa_4}{\sa_6} + \cp{\sa_1}{\sa_5}\cp{\sa_2}{\sa_4}\cp{\sa_3}{\sa_6} - 2\cp{\sa_1}{\sa_5}\cp{\sa_2}{\sa_6}\cp{\sa_3}{\sa_4}\\[2mm]
		 +&\cp{\sa_1}{\sa_6}\cp{\sa_2}{\sa_3}\cp{\sa_4}{\sa_5} + \cp{\sa_1}{\sa_6}\cp{\sa_2}{\sa_4}\cp{\sa_3}{\sa_5} - 2\cp{\sa_1}{\sa_6}\cp{\sa_2}{\sa_5}\cp{\sa_3}{\sa_4}.
	}
	Now one can simply check each summand by counting the cycles of the corresponding chord diagram, while keeping in mind that \textit{only} the bicoloured cycles with chords and $j$-coloured base edges are counted. For example, in the first term each factor corresponds to a chord $(1,2)$, $(3,4)$, $(5,6)$, each spanning exactly one of the $j$-coloured base edges. Hence, there are three such cycles and $(-2)^{c_2^j(D)-1} = 4$.
\end{example}

	The obvious next questions is now: Can we find such an identity for all partitions? Indeed, we can.

% Lemma 2
\begin{lemma}\label{lemma_part_2}
	Let $\cE \in \cP(E_{D_0}^1)$ be any partition of $1$-coloured base edges of a diagram $D_0\in \cD_0^n$ and $\sP_2$ the corresponding word pairs as above. Then
	\al{
		\sum_{ (\su,\sv)\in \sP_2 }\nquad \sgn_{\cE}(\su,\sv) \nquad \prod_{ (\su',\sv')\in\lambda_{\cE}(\su,\sv) } \nqquad\cp{\su'}{\sv'}
		= (-1)^{1-|\cE|}(-\kp)^{|\cE|-N}\nquad\sum_{D\in\cD|_{\cE}^0} (-2)^{c_2^2(D)-1} \nquad\prod_{(u,v)\in E_D^0} \nquad\cp{\sa_u}{\sa_v},		
	}
	where $\cD|_{\cE}^0\subset \cD_{\Gamma}^0 \simeq \cD_N^n$ is the subset of complete chord diagrams with base cycles given by $n$ (and vertices labelled by edges of $\Gamma$) that is restricted by demanding that all chords of a diagram can only connect vertices that lie within the same part of $\cE$.
\end{lemma}
\begin{proof}
	Consider again the word pair $( \sx_{11}\dotsm\sx_{1N}, \sx_{21}\dotsm\sx_{2N} )$. The letter pairs $(\sx_{1i}, \sx_{2i})$ correspond to $2$-coloured base edges, so the $1$-coloured base edges correspond to pairs $(\sx_{1(i+1)}, \sx_{2i})$ for $i \neq n_1, n_1+n_2, \dotsc, N$ as well as $(\sx_{11}, \sx_{2n_1})$, $(\sx_{1(n_1+1)}, \sx_{2(n_1+n_2)})$ etc. due to cyclicity in each base cycle. With this we can represent the partitions of $E_{D_0}^1$ by partitions of $\{1,\dots, N\}$.\\
	
	Assume at first that there is a single base cycle with $n_1=N$ and the partition has two parts, $I = \{i_1,\dotsc, i_{l_1}\} \subset \{1,\dots, N\}$ and $J = \{j_1,\dotsc, j_{l_2}\} \subset \{1,\dots, N\}$ with $j_{l_2} = N$. The extension to the general case is quite straightforward and discussed further below. Now look again at word pairs of the form
	\al[*]{
		( \sx_{(1+t_1)1}\dotsm\sx_{(1+t_N)N},\ \sx_{(2-t_1)1}\dotsm\sx_{(2-t_N)N} )
	}
	summed over all $N$-tuples in $\cT=\{ t\in \{0,1\}^N \ | \ t_1=0\}$. The map $\lambda_{\cE}$ restricts which tuples are permitted in the sum and describes how the remaining word pairs have to be split up. The only word pair that always yields a nonempty set under $\lambda_{\cE}$ is that of $t=(0,\dotsc,0)$ where one finds
	\al{\nonumber
		\lambda_{\cE}( \sx_{11}\dotsm&\sx_{1N}, \sx_{21}\dotsm\sx_{2N} )\\
			&= \{ ( \sx_{1(i_1+1)}\dotsm\sx_{1(i_{l_1}+1)}, \sx_{2i_1}\dotsm\sx_{2i_{l_1}} ), ( \sx_{1(j_1+1)}\dotsm\sx_{1(j_{l_2}+1)}, \sx_{2j_1}\dotsm\sx_{2j_{l_2}} )\label{eq_lemma2_lambda}
	}
	with the cyclic identification $\sx_{1(N+1)} = \sx_{11}$ understood. By construction both words in each pair have the same length, $l_1$ and $l_2$ respectively. Moreover, we can note that regardless of the specific partition the same permutation applied to both words of the concatenated pair
	\al[*]{
		( \sx_{1(i_1+1)}\dotsm\sx_{1(i_{l_1}+1)} \sx_{1(j_1+1)}\dotsm\sx_{1(j_{l_2}+1)} , \sx_{2i_1}\dotsm\sx_{2i_{l_1}}\sx_{2j_1}\dotsm\sx_{2j_{l_2}} )
	}
	returns $( \sx_{12}\dotsm\sx_{1N}\sx_{11}, \sx_{21}\dotsm\sx_{2N} )$, so here $\sgn_{\cE}(( \sx_{11}\dotsm\sx_{1N}, \sx_{21}\dotsm\sx_{2N} ) = (-1)^{N-1}$. 
	
	Next we need to study what happens for different word pairs, i.e. if the letter pairs $(\sx_{1r}, \sx_{2r})$ are exchanged for all $r$ in another subset $R \subset \{2,\dotsc N\}$. If $r$ and $r-1$ are both in $I$ or both in $J$ then the swap of $\sx_{1r}$ and $\sx_{2r}$ results in word pairs that still have equal length words since $\sx_{1r}$ is contained in the same word pair as $\sx_{2r}$. If $r$ and $r-1$ are not in the same part then we find that exchange of any single letter pair $(\sx_{1r}, \sx_{2r})$ will lead to words of different lengths in each pair such that the term does not contribute. Hence, each exchange of a letter pair $(\sx_{1r}, \sx_{2r})$ with $r\in I$ and $r-1\notin I$ will require another exchange of $(\sx_{1s}, \sx_{2s})$ with a suitable $s\in J$ to compensate and return word pairs with non-vanishing contribution. Here we need to start distinguishing between different types of partitions.\\

	First, let $I$ and $J$ be sets of consecutive numbers (counting $N$ and $1$ as such). Then there are only two $r\in\{1,\dotsc, N\}$ such that $r$ and $r-1$ are in different parts. Since only word pairs in which either both or neither are exchanged contribute one finds that exactly half of all word pairs in $\sP_2$ yield non-empty sets of pairs under $\lambda_{\cE}$. Then the sum
	\al[*]{
		\sum_{ (\su,\sv)\in \sP_2 }\nquad \sgn_{\cE}(\su,\sv) \nquad \prod_{ (\su',\sv')\in\lambda_{\cE}(\su,\sv) } \nqquad\cp{\su'}{\sv'}
	}
	contains $2^{N-2}$ terms that decompose into two factors with $2^{l_1-1}$ and $2^{l_2-1}$ terms corresponding to the two parts. Permutations with signum $(-1)^{l_1-1}$ and $(-1)^{l_2-1}$ can be used (analogous to the discussion of the sign above) to align the original letter pairs (corresponding to $2$-coloured base edges) in each word pair. Then each such factor can be rewritten with lemma \ref{lemma_fullpart}, where one interprets it as arising from a certain smaller chord diagram base cycle. That cycle, say for the part $I$, results from contraction of the path that consists all $1$-coloured base edges represented by the integers in $J$ as well as the $2$-coloured base edges in between these (consecutive) $1$-coloured edges to a single $2$-coloured base edge. Any pair of diagrams built on these smaller base cycles corresponds to a larger diagram with the original base cycle that one finds by simply cutting the contracted $2$-coloured base edge in each diagram and gluing them together. The number of $2$-coloured cycles is almost additive but the cutting removes one cycle in each diagram and restores only one when gluing them together. Hence, $c_2^2(D_I)-1 + c_2^2(D_J)-1 = c_2^2(D_{IJ})-1$.
	
	This straightforwardly extends to partitions with any number of parts, as long as each consists of consecutive base edges, and one finds
	\al{\nonumber
		\sum_{ (\su,\sv)\in \sP_2 }\nquad \sgn_{\cE}(\su,\sv) & \nquad \prod_{ (\su',\sv')\in\lambda_{\cE}(\su,\sv) } \nqquad\cp{\su'}{\sv'}\\\nonumber
		&= (-1)^{1-N} \prod_{i=1}^{|\cE|} \Big( (-1)^{l_i-1} (-\kp)^{1-l_i} \sum_{D_i\in\cD_{\Gamma,i}^0} (-2)^{c_2^2(D_i)-1} \nquad\prod_{(u,v)\in E_{D_i}^0} \cp{\sa_u}{\sa_v} \Big)\\[3mm]
		&= (-1)^{1-|\cE|}(-\kp)^{|\cE|-N}\nquad\sum_{D\in\cD|_{\cE}^0} (-2)^{c_2^2(D)-1} \nquad\prod_{(u,v)\in E_D^0} \nquad\cp{\sa_u}{\sa_v}
	}
	where $l_1,\dotsc, l_{|\cE|}$ with $l_1+\dotsc+\l_{|\cE|} = N$ are the cardinalities of each part. This even extends further to partitions like $\{ \{1\}, \{3\}, \{2,4\} \}$ where the part $\{2,4\}$ does not contain consecutive base edges initially but $2$ and $4$ become consecutive after factoring out the terms (contracting the base edges) corresponding to $1$ and $3$.\\
	
	Next we look at the exact opposite case, i.e. we assume that $I$ and $J$ do not contain any consecutive numbers at all. Note that then both parts need to have the same cardinality $|I|=|J|=N/2$ and $N$ has to be even. The contributing word pairs can be found by considering all possible choices of $k\leq N/2-1$ index swaps out of the set that contains $1$ (which is kept fixed) together with all possible choices of the same number of indices from the other set (in which all $N/2$ elements are permitted). The number of such exchanges can be counted with Vandermonde's identity to be
	\al{
		\sum_{k=0}^{N/2-1} \binom{N/2-1}{k}\binom{N/2}{k} = \binom{N-1}{N/2-1} = \frac{1}{2}\binom{N}{N/2}.
	} 
	The sum containing these terms does not factorise, but we can reduce it to a sum of expressions corresponding to $N-1$, allowing for proof by induction. Choose one of the two parts and expand the corresponding factor of each summand analogously to \refeq{eq_Dodgson_exp}. By construction the first term on the r.h.s. cannot exist in these expansions, since $\sx_{11}$ and $\sx_{21}$ belong to different word pairs. The sum contains fewer terms but the principle is the same: Suitable permutation within the remaining word pair allows us to interpret it as associated to a diagram that in turn resulted from addition of a chord corresponding to the removed letter pair. Hence, we only pick up an overall factor of $-\kp$ and can collect coefficients of each Dodgson polynomial $\cp{\sx_{1r}}{\sx_{2s}}$. By simply counting how often a given letter pair is or is not involved in an exchange one finds that one can collect terms into groups of
	\al{
		\binom{N-2}{N/2-1} = \binom{N-2}{(N-2)/2}
	}
	which is exactly twice the number of possible exchanges we would have for $N-2$. The coefficient of each $\cp{\sx_{1r}}{\sx_{2s}}$ corresponds to the sum in \refeq{eq_lemma2_lambda} but for a smaller diagram with $N'=N-1$ and a corresponding smaller partition. For the small cases of $N=2$ and $N=4$ the reduction already yields factorising expressions (see example \ref{ex_lemma_part_2}). For larger $N$ that is generally not the case and since $N-1$ is odd it also cannot belong to the case we discussed here. Instead, what happens is a partial factorisation that allows us to collect the $2$, $6$, $20$, $\dotsc$ terms into $1$, $3$, $10$, $\dotsc$ pairs which correspond to a non-factorising partition with a total cardinality of $N-2$. The corresponding partition consists of one part in which all elements are still non-consecutive and one part that contains only exactly one pair of consecutive numbers. Then the reduction process goes through for any such partition with mixed consecutive and non-consecutive base edges, even for more than two parts. If there are $k$ pairs of consecutive base edges in one part then this simply yields $2^{k-1}$ terms which correspond to a subset of the possible word pairs resulting from some smaller diagram -- but it is not the full subset needed for the factorisation seen above.\\
	
	Finally, all of this goes through for any number of base cycles without much change. The only difference is in which base edges are viewed as consecutive. For example, for a diagram with two base cycles of size $n_1$ and $n_2$ with $n_1+n_2=N$ one has $(1,n_1)$ and $(n_1+1, N)$ as consecutive pairs, but neither $(1,N)$ nor $(n_1,n_1+1)$.
\end{proof}

% example
\begin{example}\label{ex_lemma_part_2}
	Consider an empty chord diagram $D\in \cD^4_0$ on a single base cycle with $8$ vertices labelled $1$--$8$. Let $2$ be the colour of the base edges $(1,2), (3,4), (5,6), (7,8)$ and $1$ the colour of $(8,1), (2,3), (4,5), (6,7)$. The word pairs in $\sP_2$ are, up to possible permutations,
	\al[*]{
		&( \sa_1\sa_3\sa_5\sa_7, \sa_2\sa_4\sa_6\sa_8 )\quad
		&( \sa_1\sa_3\sa_5\sa_8, \sa_2\sa_4\sa_6\sa_7 )\quad
		&( \sa_1\sa_3\sa_6\sa_7, \sa_2\sa_4\sa_5\sa_8 )\quad
		&( \sa_1\sa_4\sa_5\sa_7, \sa_2\sa_3\sa_6\sa_8 )\; \\[2mm]
		&( \sa_1\sa_3\sa_6\sa_8, \sa_2\sa_4\sa_5\sa_7 )\quad
		&( \sa_1\sa_4\sa_5\sa_8, \sa_2\sa_3\sa_6\sa_7 )\quad
		&( \sa_1\sa_4\sa_6\sa_7, \sa_2\sa_3\sa_5\sa_8 )\quad
		&( \sa_1\sa_4\sa_6\sa_8, \sa_2\sa_3\sa_5\sa_7 ).
	}
	The partitions with one part are very similar to those in example \ref{ex_lem1} but a bit too large already to sensibly write them here in their fully expanded form. With two parts there are three types of partitions. Firstly, if $\cE = \{E_1, E_2\}$ with $|E_1| = 3$, $|E_2|=1$, then the factorisation is obvious. For example, for $\cE = \{ \{(8,1), (2,3), (4,5)\},  \{(6,7)\}$ one has
	\al[*]{
		- \cp{\sa_6}{\sa_7}\Big( 
							\cp{ \sa_1\sa_3\sa_5}{\sa_2\sa_4\sa_8} + 
							\cp{ \sa_1\sa_4\sa_5}{\sa_2\sa_3\sa_8} + 
							\cp{ \sa_1\sa_3\sa_8}{\sa_2\sa_4\sa_5} + 
							\cp{ \sa_1\sa_4\sa_8}{\sa_2\sa_3\sa_5}
						\Big).
	}
	Similarly, for a partition like $\cE = \{ \{(8,1), (2,3)\},  \{(4,5), (6,7)\}$ one also finds a factorisation since the four terms one gets are
	\al{
		\cp{ \sa_1\sa_3}{\sa_2\sa_8} \cp{ \sa_5\sa_7}{\sa_4\sa_6} + 
		\cp{ \sa_1\sa_3}{\sa_2\sa_8} \cp{ \sa_6\sa_7}{\sa_4\sa_5} + 
		\cp{ \sa_1\sa_8}{\sa_2\sa_3} \cp{ \sa_6\sa_7}{\sa_4\sa_5} +
		\cp{ \sa_1\sa_8}{\sa_2\sa_3} \cp{ \sa_5\sa_7}{\sa_4\sa_6}.
	}
	The non-factorising partition $\cE = \{ \{(8,1), (4,5)\},  \{(2,3), (6,7)\}$ yields three terms that we can still simply expand explicitly:
	\al{\nonumber
		\big(-\kp\big)^2&  \Big(
						\cp{ \sa_1\sa_5}{\sa_4\sa_8} \cp{ \sa_3\sa_7}{\sa_2\sa_6} +
					 	\cp{ \sa_1\sa_8}{\sa_4\sa_5} \cp{ \sa_3\sa_6}{\sa_2\sa_7} + 
						\cp{ \sa_1\sa_4}{\sa_5\sa_8} \cp{ \sa_6\sa_7}{\sa_2\sa_3}
					\Big)\\[2mm]\nonumber
		=& 	\Big( \cp{ \sa_1}{\sa_4}\cp{ \sa_5}{\sa_8} - \cp{ \sa_1}{\sa_8}\cp{ \sa_4}{\sa_5} \Big)\Big( \cp{ \sa_2}{\sa_3}\cp{ \sa_6}{\sa_7} - \cp{ \sa_2}{\sa_7}\cp{ \sa_3}{\sa_6} \Big)\\\nonumber
		& + 	\Big( \cp{ \sa_1}{\sa_4}\cp{ \sa_5}{\sa_8} - \cp{ \sa_1}{\sa_5}\cp{ \sa_4}{\sa_8} \Big)\Big( \cp{ \sa_2}{\sa_3}\cp{ \sa_6}{\sa_7} - \cp{ \sa_2}{\sa_6}\cp{ \sa_3}{\sa_7} \Big)\\\nonumber
		& +	\Big( \cp{ \sa_1}{\sa_5}\cp{ \sa_4}{\sa_8} - \cp{ \sa_1}{\sa_8}\cp{ \sa_4}{\sa_5} \Big)\Big( \cp{ \sa_2}{\sa_6}\cp{ \sa_3}{\sa_7} - \cp{ \sa_2}{\sa_7}\cp{ \sa_3}{\sa_6} \Big)\\[2mm]\nonumber
		=& (-1)^3(-2)^1 \Big( 	
					\cp{ \sa_1}{\sa_4}\cp{ \sa_2}{\sa_3}\cp{ \sa_5}{\sa_8}\cp{ \sa_6}{\sa_7} + 
					\cp{ \sa_1}{\sa_5}\cp{ \sa_2}{\sa_6}\cp{ \sa_3}{\sa_7}\cp{ \sa_4}{\sa_8}\\\nonumber
		& \hphantom{=(-1)^3(-2)^1\Big( \cp{ \sa_1}{\sa_4}\cp{ \sa_2}{\sa_3}\cp{ \sa_5}{\sa_8}\cp{ \sa_6}{\sa_7}} 	
				 +	\cp{ \sa_1}{\sa_8}\cp{ \sa_2}{\sa_7}\cp{ \sa_3}{\sa_6}\cp{ \sa_4}{\sa_5} \Big)\\\nonumber
		&+ (-1)^3(-2)^0 \Big( 	
					\cp{ \sa_1}{\sa_4}\cp{ \sa_2}{\sa_6}\cp{ \sa_3}{\sa_7}\cp{ \sa_5}{\sa_8} + 					
					\cp{ \sa_1}{\sa_4}\cp{ \sa_2}{\sa_7}\cp{ \sa_3}{\sa_6}\cp{ \sa_5}{\sa_8}\\\nonumber
		& \hphantom{(-1)^3(-2)^0 \Big(} 
		+			\cp{ \sa_1}{\sa_5}\cp{ \sa_2}{\sa_3}\cp{ \sa_4}{\sa_8}\cp{ \sa_6}{\sa_7} + 
					\cp{ \sa_1}{\sa_5}\cp{ \sa_2}{\sa_7}\cp{ \sa_3}{\sa_6}\cp{ \sa_4}{\sa_8}\\
		& \hphantom{(-1)^3(-2)^0 \Big(}
		+			\cp{ \sa_1}{\sa_8}\cp{ \sa_2}{\sa_3}\cp{ \sa_4}{\sa_5}\cp{ \sa_6}{\sa_7} +
					\cp{ \sa_1}{\sa_8}\cp{ \sa_2}{\sa_6}\cp{ \sa_3}{\sa_7}\cp{ \sa_4}{\sa_5}
					\Big).
	}
	With a quick drawing one can now check that the chord diagrams corresponding to these terms are as expected and that the number of cycles is indeed correct. Finally, expanding only one of the two polynomials in each summand leads to the reduction from the proof of lemma \ref{lemma_part_2}:
	\al[*]{
		&	\cp{ \sa_1}{\sa_4}\cp{ \sa_5}{\sa_8}\big( \cp{ \sa_3\sa_7}{\sa_2\sa_6} + \cp{ \sa_3\sa_6}{\sa_2\sa_7} \big)\\
		+&	\cp{ \sa_1}{\sa_5}\cp{ \sa_4}{\sa_8}\big( \cp{ \sa_2\sa_7}{\sa_6\sa_3} + \cp{ \sa_2\sa_3}{\sa_6\sa_7} \big)\\
		+&	\cp{ \sa_1}{\sa_8}\cp{ \sa_4}{\sa_5}\big( \cp{ \sa_2\sa_6}{\sa_7\sa_3} + \cp{ \sa_2\sa_3}{\sa_7\sa_6} \big).
	}
\end{example}

	The final ingredient for the proof of this chapter's main theorem is an identity allowing summation of Stirling numbers of the second kind $S(k,l)$. They count the ways to partition a set of $k$ elements into $l$ non-empty sets. To prove it we need a certain identity relating Stirling numbers and the classical polylogarithm. While the literature contains a number of well known identities that do so, they are all either similar but not obviously equivalent to the one we need, or appear without proof. Moreover, the commonly cited references (e.g. \cite{abramowitz+stegun, Stanley_EnumComb, knuth_1998}, among many others) all appear to cite each other or unavailable older literature, so it may actually be somewhat elucidating to derive everything we need ourselves.

\begin{proposition}\label{prop_polylog}
	Let
	\al{
		\Li_s(z) = \sum_{l=1}^{\infty}\frac{z^l}{l^s}\qquad |z|<1, \ s\in \ZZ
	}
	be the classical polylogarithm and $S(k,l)$ be the Stirling number of the second kind. Then
	\al{
		\Li_{-k+1}(z) = (-1)^k\sum_{l=1}^k S(k,l) \frac{(l-1)!}{(z-1)^l} 
	}
	for integers $k\geq 2$.
\end{proposition}
\begin{proof}
	For $k=2$ the r.h.s. is
	\al{
		\frac{1}{z-1} + \frac{1}{(z-1)^2} = \frac{z}{(1-z)^2} = z\partial_z\frac{1}{1-z}
				= z\partial_z\sum_{l=0}^{\infty} z^l
		 		= \sum_{l=1}^{\infty}l z^l = \Li_{-1}(z).
	}
	Now proceed by induction	
	\al{\nonumber
		\Li_{-k+1}(z) = z\partial_z \Li_{-k+2}(z) 
				&= (-1)^{k-1}\sum_{l=1}^{k-1} S(k-1,l) z\partial_z \frac{1}{(z-1)^l}(l-1)!\\
				&= (-1)^k\sum_{l=1}^{k-1} S(k-1,l) \frac{z}{(z-1)^{l+1}} l!,
	}
	and use partial fraction decomposition to find
	\al{
		S(k-1,l) \frac{z}{(z-1)^{l+1}}l!& = lS(k-1,l) \frac{(l-1)!}{(z-1)^l} + S(k-1,l)\frac{l!}{(z-1)^{l+1}}\label{eq_lem_auxA}
	}	
	Using the recurrence relation $S(k,l) = S(k-1,l-1) + l(S(k-1,l)$ the first term is further rewritten as
	\al{
		lS(k-1,l) \frac{(l-1)!}{(z-1)^l} = S(k,l)\frac{(l-1)!}{(z-1)^l} - S(k-1,l-1)\frac{(l-1)!}{(z-1)^l}\label{eq_lem_auxB}
	}
	In the sum one now has a telescopic cancellation involving the second terms of eqs. (\ref{eq_lem_auxA}) and (\ref{eq_lem_auxB}). The only remaining terms are
	\al[*]{
		\frac{S(k-1,0)}{z-1} = 0 \qquad\text{ and }\qquad S(k-1,k-1)\frac{(k-1)!}{(z-1)^k} = S(k,k)\frac{(k-1)!}{(z-1)^k},
	}
	as well as the first part of the r.h.s. of \refeq{eq_lem_auxB} summed up to $l=k-1$, such that overall
	\al{\nonumber
		\Li_{-k+1}(z) &= (-1)^k\sum_{l=1}^k S(k,l)\frac{(l-1)!}{(z-1)^l}.
	}
\end{proof}

%% LEMMA
\begin{lemma}\label{lemma_maintheo}
	Let $S(k,l)$ be the Stirling number of the second kind. Then
	\al{\nonumber
		\sum_{l=1}^{k} S( k, l ) (-1)^l(l+1)!  = (-2)^k\qquad \forall k \geq 1.
	}
\end{lemma}
\begin{proof}
	For $k=1$ the claim is checked directly. For $k\geq 2$ we use the identity derived for the polylogarithm in proposition \ref{prop_polylog} and note that a change of the argument allows us to write
	\al{
		(-1)^k\Li_{-k+1}\Big(1+\frac{1}{z}\Big) = \sum_{l=1}^k S(k,l)z^l(l-1)!
	}
	with $z < -1$. Now let 
	\al{\nonumber
		L(z) \defeq \sum_{l=1}^{k} S( k, l ) z^l(l+1)! 
			&= z\partial^2_z z\sum_{l=1}^{k} S( k, l ) z^l(l-1)!\\
			&= (-1)^k z\partial^2_z z \Li_{-k+1}\Big(1+\frac{1}{z}\Big).
	}
	Computing the derivative one finds
	\al{
		L(z) = \frac{(-1)^k}{(z+1)^2}\bigg( \Li_{-k-1}\Big(1+\frac{1}{z}\Big) - \Li_{-k}\Big(1+\frac{1}{z}\Big) \bigg).
	}
	Both polylogarithms start with terms linear in $(z+1)/z$, yielding divergences when evaluating at $z=-1$, but upon closer inspection we see that they precisely cancel each other. With $z<-1$ one has $|1+1/z|<1$ such that we are able to employ the classical sum representation of the polylogarithm, of which only the first two terms are of interest to us:
	\al{\nonumber
		L(z) &= \frac{(-1)^k}{(z+1)^2}\bigg( \sum_{t=1}^{\infty} t^{k+1}\left(\frac{z+1}{z}\right)^t
								-\sum_{t=1}^{\infty} t^k\left(\frac{z+1}{z}\right)^t \bigg)\\[3mm]\nonumber
			&= \frac{(-1)^k}{(z+1)^2}\bigg( \frac{z+1}{z} + 2^{k+1}\left(\frac{z+1}{z}\right)^2
								- \frac{z+1}{z} - 2^k\left(\frac{z+1}{z}\right)^2 
								+ \cO\left(  \left(\frac{z+1}{z}\right)^3 \right) \bigg)\\[3mm]
			&= (-2)^k \bigg( \frac{1}{z^2} + \frac{1}{(z+1)^2}\cO\left(\left(\frac{z+1}{z}\right)^3\right) \bigg).
	}
	Now we can safely take the limit $z \to -1$ to find
	\al{
		\sum_{l=1}^{k} S( k, l ) (-1)^l(l+1)! = L(-1) = (-2)^k.
	}
\end{proof}

	% THEOREM
	\paragraph{Proof of Theorem \ref{theo_main}.} %\begin{proof}[\textbf{Proof of theorem \ref{theo_main}}]
	First, use lemma \ref{lemma_part_2} to rewrite the partition polynomial as
	\al{\nonumber
		Z_{\Gamma}^0 &= \sum_{\cE \in \cP(E_D^1)} (-\kp)^{N-|\cE|} (|\cE|+1)! \sum_{ (\su,\sv)\in \sP_2 } \sgn_{\cE}(\su,\sv) \nquad
			\prod_{ (\su',\sv')\in\lambda_{\cE}(\su,\sv) } \nqquad\cp{\su'}{\sv'}\\[3mm]\nonumber
			&= \sum_{\cE \in \cP(E_D^1)} (-1)^{|\cE|+1}(|\cE|+1)!\sum_{D\in\cD|_{\cE}^0} (-2)^{c_2^2(D)-1} \nquad\prod_{(u,v)\in E_D^0} \cp{\sa_u}{\sa_v}.
	}
	The sum already contains $c_2^2(D)$, the number of $2$-coloured cycles. Regarding cycles of the other colour we can make the following observation: In each diagram with $c_2^1(D)\leq N$ the $1$-coloured cycles can themselves be interpreted as a partition of $E_D^1$ in which each part is given by the base edges connected to each other by chords. The diagrams in $\cD|_{\cE}^0$ can only have chords connecting base edges within the same part of $\cE$, so each part in the partition given by the $1$-coloured cycles has to be a subset of a part of $\cE$. Counting the number of ways of partitioning the $c_2^1(D)$ cycles of a given diagram into partitions with $|\cE|$ parts (i.e. counting the number of partitions $\cE$ with a certain number of parts such that $\cD|_{\cE}^0$ contains the given diagram $D$) one finds precisely the Stirling numbers of the second kind $S(c_2^1(D), |\cE|)$. Using this, we can exchange summation over diagrams and partitions and find
	\al{\nonumber
		Z_{\Gamma}^0 &= \sum_{\cE \in \cP(E_D^1)} (-1)^{|\cE|+1}(|\cE|+1)!\sum_{D\in\cD|_{\cE}^0} (-2)^{c_2^2(D)-1}
							\nquad\prod_{(u,v)\in E_D^0} \cp{\sa_u}{\sa_v}\\[3mm]\nonumber
			&= \frac{1}{2}\sum_{D\in\cD^n_N} (-2)^{c_2^2(D)}\bigg( \prod_{(u,v)\in E_D^0} \cp{\sa_u}{\sa_v}\bigg) \sum_{l=1}^{c_2^1(D)} S( c_2^1(D), l ) (-1)^l(l+1)!.
	}
	Now lemma \ref{lemma_maintheo} is applied to evaluate the sum to $(-2)^{c_2^1(D)}$, which finishes the proof.
	\qed

%
%
%
% SECOND SUMMATION
\subsection{The second summation theorem}\label{sec_inc}

	Now that $Z_{\Gamma}^0$ is well understood we can proceed to the more complicated $Z_{\Gamma}^1$. Contrary to $Z_{\Gamma}^0$ they contain not only the cycle polynomials, but also $x_{\Gamma}^e$, which we had defined in \refeq{eq_CP_eval_def}. We begin by analysing these polynomials and in particular their products a bit further. Building on this we will then find that the summation theorem from the previous section can be generalised rather straightforwardly to the following result.

	% theorem	
	\begin{theorem}\label{theo_main_2}
		\al{
			Z_{\Gamma}^1 = \frac{1}{2} \sum_{D\in \cD_{\Gamma}^1} (-2)^{\tilde c(D)} \bigg(\prod_{ (u,v)\in E_D^0 } \cp{u}{v}\bigg) \prod_{ w\in V_D^{(2)} } x_{\Gamma}^w.
			 \label{eq_maintheo_2}
		}
	\end{theorem}

%%% PRODUCTS OF x
	\subsubsection{The polynomial $x_{\Gamma}^w$}
	
	\noindent The first step to prove this theorem is getting a better understanding of the polynomials $x_{\Gamma}^w$ and their products. We begin with some general observations about their connections to bond and spanning forest polynomials and then state the precise result that we will need in lemma \ref{lemma_prodX} below.\\
	
		Analogous to \refeq{eq_ssp_stmat} we can also write the bond polynomial as
		\al{ 
			\beta_G &= \sum_{v_1,v_2\in V_G} \vartheta_{v_1}\vartheta_{v_2} \Phi_G^{\{v_0\},\{v_1,v_2\}},
		}
		where the momenta $q_{v_i}$ are replaced with $\vartheta_{v_i} = \sum_e I_{ev_i} \xi_e$. With the definition of $\CP[G]{e}{\mu}$ as derivative of the bond polynomial w.r.t. $\xi_e^{\mu}$ (see \refeq{eq_CPdef}) one finds
		\al{
			\CP[G]{e}{\mu} 
				= \alpha_e^{-1} \sum_{v_1,v_2\in V_G} I_{ev_1} \vartheta_{v_2}^{\mu} \Phi_G^{\{v_0\},\{v_1,v_2\}}.
		}
		Then we move to the physical case, i.e. a Feynman graph $\Gamma$ in which we evaluate the formal parameters $\xi_e$ to physical momenta. For each edge there are only two vertices, namely $u_1,u_2$ with $\partial(e) = (u_1,u_2)$, such that $I_{eu_i} \neq 0$, and $I_{eu_1} = -I_{eu_2}$ for this pair. Hence, the polynomial reduces to
		\al{
			\CP{e}{\mu} = - \alpha_e^{-1} \sum_{v\in V_{\Gamma}^{ext}} q_v^{\mu} \big( \Phi_{\Gamma}^{\{v_0\},\{u_2,v\}} - \Phi_{\Gamma}^{\{v_0\},\{u_1,v\}} \big).
		}
		Accounting for cancellations between spanning forests (i.e. their corresponding monomials) that appear in both polynomials, the difference can be written as
		\al{
			\Phi_{\Gamma}^{\{v_0\},\{u_2,v\}} - \Phi_{\Gamma}^{\{v_0\},\{u_1,v\}} 
			= \Phi_{\Gamma}^{\{v_0,u_1\},\{u_2,v\}} - \Phi_{\Gamma}^{\{v_0,u_2\},\{u_1,v\}}. \label{eq_addsp}
		}
		If we now specialise to the case of only two external vertices $v_1,v_2$ (or at least only two with non-vanishing momenta), then this reduces further to
		\al{
			\CP{e}{\mu} = q^{\mu}\alpha_e^{-1} \big(  \Phi_{\Gamma}^{\{v_0\},\{u_1,v_1\}} + \Phi_{\Gamma}^{\{v_0\},\{u_2,v_2\}} - \Phi_{\Gamma}^{\{v_0\},\{u_2,v_1\}} - \Phi_{\Gamma}^{\{v_0\},\{u_1,v_2\}} \big)
		}
		In order to explain the overall sign we emphasise again that $e$ is directed from $\partial_-(e) = u_1$ to $\partial_+(e) = u_2$, and that we chose $q_{v_1} = q = -q_{v_2}$.\\
	
		By the same principle as \refeq{eq_addsp} we can explicitly remove terms that would cancel between these four summands:
		\al{\nonumber
			& \Phi_{\Gamma}^{\{v_0\},\{u_1,v_1\}} - \Phi_{\Gamma}^{\{v_0\},\{u_2,v_1\}} + \Phi_{\Gamma}^{\{v_0\},\{u_2,v_2\}}- \Phi_{\Gamma}^{\{v_0\},\{u_1,v_2\}}\\[3mm]\nonumber
			&\qquad= \Phi_{\Gamma}^{\{v_0,u_2\},\{u_1,v_1\}} - \Phi_{\Gamma}^{\{v_0,u_1\},\{u_2,v_1\}} + \Phi_{\Gamma}^{\{v_0,u_1\},\{u_2,v_2\}} - \Phi_{\Gamma}^{\{v_0,u_2\},\{u_1,v_2\}}\\[3mm]\nonumber
			&\qquad= \Phi_{\Gamma}^{\{v_0,u_2,v_2\},\{u_1,v_1\}} + \Phi_{\Gamma}^{\{v_0,u_1,v_1\},\{u_2,v_2\}} - \Phi_{\Gamma}^{\{v_0,u_2,v_1\},\{u_1,v_2\}}  - \Phi_{\Gamma}^{\{v_0,u_1,v_2\},\{u_2,v_1\}} \\[3mm]
			&\qquad= \Phi_{\Gamma}^{\{u_1,v_1\},\{u_2,v_2\}} - \Phi_{\Gamma}^{\{u_1,v_2\},\{u_2,v_1\}}. \label{eq_spshort}
		}
		This is now explicitly independent of the arbitrarily chosen vertex $v_0$. We can re-expand \refeq{eq_spshort} by including terms cancelled between the two to get
		\al{
			\Phi_{\Gamma}^{\{u_1,v_1\},\{u_2,v_2\}} - \Phi_{\Gamma}^{\{u_1,v_2\},\{u_2,v_1\}}
			= \Phi_{\Gamma}^{\{v_1\},\{u_2,v_2\}} - \Phi_{\Gamma}^{\{v_1\},\{u_1,v_2\}}.
		}
		This is now not only independent of the original arbitrary choice of $v_0$ but can actually be interpreted as Dodgson polynomials with respect to a graph matrix in which $v_1$ was removed:
		\al{
			\CP{e}{\mu} = q^{\mu}\alpha_e^{-1} \big( \cp[\Gamma,v_1]{\sa_{u_2}}{\sa_{v_2}} - \cp[\Gamma,v_1]{\sa_{u_1}}{\sa_{v_2}} \big) = q^{\mu}x_{\Gamma}^e.	\label{eq_newxpol}
		}
		
	%\vspace{5mm}
	\begin{lemma}\label{lemma_prodX}
		Let $\Gamma$ be a QED Feynman graph with only two non-zero external momenta $q_u = q = -q_v$ at vertices $u,v \in V_{\Gamma}$, and $\Gamma^{\bullet} = \Gamma|_{u=v}$. Let furthermore $e,f \in E_{\Gamma}$ be any two edges of $\Gamma$. Then
		\al{
			\alpha_e\alpha_f x_{\Gamma}^ex_{\Gamma}^f &=  \kp[\Gamma^{\bullet}] \bp{e}{f} - \kp \bp[\Gamma^{\bullet}]{e}{f}.
											\label{eq_theo}
		}
		Moreover, if $e\neq f$ this simplifies to
		\al{
				x_{\Gamma}^ex_{\Gamma}^f = -\kp[\Gamma^{\bullet}] \cp{e}{f} + \kp \cp[\Gamma^{\bullet}]{e}{f}, \label{eq_theo_rewrite}
		}
		which means that up to sign $x_{\Gamma}^e = \pm \cp[\Gamma,u]{ e }{ v }$ and the signs are such that
		\al{
			x_{\Gamma}^ex_{\Gamma}^f = - \cp[\Gamma,u]{ e }{ v }\cp[\Gamma,u]{ f }{ v }.	\label{eq_xx_dp_signs}
		}
	\end{lemma}
	\begin{proof}
		Let $a,b,c,d \in V_{\Gamma}$ be the not necessarily distinct endpoints of edges $e$ and $f$, with directions $\partial(e) = (a,b)$ and $\partial(f)=(c,d)$, and use letters $\sa \equiv \sa_a$, $\ssb \equiv \sa_b$, etc. With \refeq{eq_newxpol} the product is then
		\al{\nonumber
			\alpha_e\alpha_f x_{\Gamma}^ex_{\Gamma}^f &=
				\big(  \cp[\Gamma,u]{\ssb}{\sv} - \cp[\Gamma,u]{\sa}{\sv} \big) \big(  \cp[\Gamma,u]{\sd}{\sv} - \cp[\Gamma,u]{\ssc}{\sv} \big)\\[3mm]\nonumber
			& = \cp[\Gamma,u]{\ssb}{\sv}\cp[\Gamma,u]{\sd}{\sv} - \cp[\Gamma,u]{\sa}{\sv}\cp[\Gamma,u]{\sd}{\sv}  - \cp[\Gamma,u]{\ssb}{\sv}\cp[\Gamma,u]{\ssc}{\sv} 
												+ \cp[\Gamma,u]{\sa}{\sv}\cp[\Gamma,u]{\ssc}{\sv}\\[3mm]\nonumber
			& = \cp[\Gamma,u]{\sv}{\sv} \big( \cp[\Gamma,u]{\ssb}{\sd} - \cp[\Gamma,u]{\sa}{\sd} - \cp[\Gamma,u]{\ssb}{\ssc} + \cp[\Gamma,u]{\sa}{\ssc} \big)\\[1mm]
			&\quad - \kp\big( \cp[\Gamma,u]{\ssb\sv}{\sd\sv} - \cp[\Gamma,u]{\sa\sv}{\sd\sv} - \cp[\Gamma,u]{\ssb\sv}{\ssc\sv} + \cp[\Gamma,u]{\sa\sv}{\ssc\sv} \big).
		}
		The coefficient of $\cp[\Gamma,u]{\sv}{\sv}$ in the first summand is exactly the sum from \refeq{eq_spshort} with different labels, such that
		\al{\nonumber
			\cp[\Gamma,u]{\ssb}{\sd} - \cp[\Gamma,u]{\sa}{\sd} - \cp[\Gamma,u]{\ssb}{\ssc} + \cp[\Gamma,u]{\sa}{\ssc}
			&= \Phi_{\Gamma}^{\{u\},\{b,d\}} - \Phi_{\Gamma}^{\{u\},\{a,d\}} - \Phi_{\Gamma}^{\{u\},\{b,c\}} + \Phi_{\Gamma}^{\{u\},\{a,c\}}\\[3mm]
			&= \Phi_{\Gamma}^{\{a,c\},\{b,d\}} - \Phi_{\Gamma}^{\{b,c\},\{a,d\}}.\label{eq_dodgsonrem}
		}
		$\cp[\Gamma,u]{\sv}{\sv}$ itself is the Kirchhoff polynomial $\kp[\Gamma^{\bullet}] = \vssp$. The terms in the coefficient of $\kp$ can be interpreted as
		\al{
			\cp[\Gamma,u]{\sa\sv}{\sd\sv} = \cp[\Gamma^{\bullet},u]{\sa}{\sd},
		}
		such that they add up to
		\al{
			\Phi_{\Gamma^{\bullet}}^{\{a,c\},\{b,d\}} - \Phi_{\Gamma^{\bullet}}^{\{b,c\},\{a,d\}},
		}
		just like \refeq{eq_dodgsonrem}. After putting all of this together we have proved the first claim,
		\al{\nonumber
			\alpha_e\alpha_f x_{\Gamma}^ex_{\Gamma}^f
			&= \kp[\Gamma^{\bullet}] \big( \Phi_{\Gamma}^{\{\partial_-(e),\partial_-(f)\},\{\partial_+(e),\partial_+(f)\}} - \Phi_{\Gamma}^{\{\partial_+(e),\partial_-(f)\},\{\partial_-(e),\partial_+(f)\}} \big)\\[1mm]\nonumber
			&\  - \kp \big( \Phi_{\Gamma^{\bullet}}^{\{\partial_-(e),\partial_-(f)\},\{\partial_+(e),\partial_+(f)\}} - \Phi_{\Gamma^{\bullet}}^{\{\partial_+(e),\partial_-(f)\},\{\partial_-(e),\partial_+(f)\}}  \big)\\[3mm]
			&=  \kp[\Gamma^{\bullet}] \bp{e}{f} - \kp \bp[\Gamma^{\bullet}]{e}{f}.
		}
		For the second claim we simply remember \refeq{eq_lemma_bpcp},
		\al[*]{
			\bp{e}{f} = -\alpha_e\alpha_f \cp{e}{f}\qquad \text{ for all } e\neq f
		}
		and divide by $\alpha_e\alpha_f$ on both sides. For the final claim we return from the notation with $\Gamma^{\bullet}$ to Dodgson polynomials. Then we have
		\al{\nonumber
			x_{\Gamma}^ex_{\Gamma}^f 	&= -\kp[\Gamma^{\bullet}] \cp{e}{f} + \kp \cp[\Gamma^{\bullet}]{e}{f}\\\nonumber
						&\nquad \Longleftrightarrow\\
			\cp[\Gamma,u]{\sv}{\sv} \cp[\Gamma,u]{\sa_e}{\sa_f} + x_{\Gamma}^ex_{\Gamma}^f	&= \kp \cp[\Gamma,u]{\sa_e \sv}{\sa_f \sv}			
		}
		and the nature of the $x_{\Gamma}^e$ becomes obvious from a comparison with the Dodgson identity in \refeq{eq_DodgsonID} or \refeq{eq_DI_rewrite}.
	\end{proof}

	%SUM
	\subsubsection{The sum over $\cD_{\Gamma}^1$}
	
	Now that we know that the additional polynomials $x_{\Gamma}^e$ are also just Dodgson polynomials it seems reasonable to think that the ideas used for the previous summation can also be used here. We find that this is indeed the case, but there are some complications that we need to consider first.\\
	
	Note that $\cD_{\Gamma}^1$ has
	\al{
		|\cD_{\Gamma}^1| = \binom{2N}{2} (2N-3)!! = \frac{2N (2N-1)!}{2 (2N-2)!} (2N-3)!! = N (2N-1)!!
	}
	elements. They can be sorted into $(2N-1)!!$ groups of $N$ diagrams, each of which corresponds to a diagram $D\in \cD_{\Gamma}^0$ and all $N$ possible choices to remove one chord from it. Hence, a sum over $\cD_{\Gamma}^1$ can be split into a double sum over $\cD_{\Gamma}^0$ and chords of each diagram. The addition of the final chord always raises the total cycle number by one, by removing the tricoloured cycle to add one bicoloured cycle of each colour. With polynomials one has
	\al{\nonumber
		\sum_{D\in \cD_{\Gamma}^1} (-2)^{\tilde c(D)} & \Big( \prod_{(u,v)\in E_D^0}\nquad \cp{\sa_u}{\sa_v} \Big) \prod_{w\in V_D^{(2)}} x_{\Gamma}^w\\[2mm]\nonumber
			&= \sum_{D\in \cD_{\Gamma}^0} (-2)^{\tilde c(D)-1} \Big( \prod_{(u,v)\in E_D^0}\nquad \cp{\sa_u}{\sa_v} \Big) \sum_{(u,v)\in E_D^0} \frac{x_{\Gamma}^ux_{\Gamma}^v}{\cp{\sa_u}{\sa_v}}\\[2mm]
			&= -\sum_{D\in \cD_{\Gamma}^0} (-2)^{\tilde c(D)-1} \Big( \prod_{(u,v)\in E_D^0}\nquad \cp{\sa_u}{\sa_v} \Big) \sum_{(u,v)\in E_D^0} \frac{ \cp{\sa_u}{\sy}\cp{\sa_v}{\sy} }{\cp{\sa_u}{\sa_v}}.
			\label{eq_sumD0D1}
	}
	Here and for the rest of this section we still assume that $\Gamma$ has two external vertices, say $x,y\in V_{\Gamma}$, all Dodgson polynomials are with respect to the vertex $x$ with the incoming momentum $q_x= q$, i.e. $\cp{\sa_u}{\sa_v} \equiv \cp[\Gamma,x]{\sa_u}{\sa_v}$, and $\sy$ is the letter associated to the other vertex with the outgoing momentum.
	
	Define the set of diagrams $\cD|_{\cE}^1 \subset \cD_{\Gamma}^1$ restricted by a partition analogously to the previous case $\cD|_{\cE}^0$. Chords are only allowed between base edges belonging to the same part and the two free vertices are treated as if they had a chord between them. In other words, a diagram $D \in \cD_{\Gamma}^1$ is in $\cD|_{\cE}^1$ if and only if the corresponding diagram $D' \in \cD_{\Gamma}^0$ (resulting from addition of the missing chord) is in $\cD|_{\cE}^0$.
	
	The next lemma is the analogue of lemmata \ref{lemma_fullpart} and \ref{lemma_part_2}. Since the idea behind the proof is very similar we directly combine them into one. 
		
	% Lemma
	\begin{lemma}\label{lemma_z1_lemma}
		Let $\cE \in \cP(E_D^1)$ be any partition of $1$-coloured base edges. Then
		\al{\nonumber
			\sum_{ (\su,\sv)\in \sP_2 }\nquad \sgn_{\cE}(\su,\sv) & \nquad \prod_{ (\su',\sv')\in\lambda_{\cE}(\su,\sv) } 
						\nqquad\cp{\su'}{\sv'}\Big( - |\cE| \kp[\Gamma^{\bullet}] + \kp \nqquad \sum_{ (\su',\sv')\in\lambda_{\cE}(\su,\sv) } \frac{\cp{\su'\sy}{\sv'\sy}}{\cp{\su'}{\sv'}} \Big)\\[3mm]
				& = (-1)^{1-|\cE|}(-\kp)^{|\cE|-N}\nquad\sum_{D\in\cD|_{\cE}^1} (-2)^{c_2^2(D)} \Big(\nquad\prod_{(u,v)\in E_D^0} \nquad\cp{\sa_u}{\sa_v}\Big)\!\! \prod_{w\in V_D^{(2)}}\nquad x_{\Gamma}^w,
		}
	\end{lemma}
	\begin{proof}
		Let $\sw_i = \sx_{i1}\dotsm \sx_{iN}$, $i=1,2$ be the two words from \refeq{eq_Dodgson_exp} in the proof of lemma \ref{lemma_fullpart}. Append the letter $\sy$ to the front of both words and consider again the expansion
		\al{\nonumber
			\kp^2 \cp{\sy\sw_1}{\sy\sw_2}
			&= \kp \cp{\sy}{\sy}\cp{\sw_1}{\sw_2} + \kp\sum_{i=1}^N (-1)^i \cp{\sx_{1i}}{\sy} \cp{\sy\sx_{11} \dotsm \hat\sx_{1i} \dotsm \sx_{1N}}{\sx_{21}\dotsm\sx_{2N}}\\[2mm]
			&= \kp[\Gamma^{\bullet}]\  \kp\cp{\sw_1}{\sw_2} - \sum_{i,j=1}^N (-1)^{i+j} \cp{\sx_{1i}}{\sy} \cp{\sx_{2j}}{\sy} \cp{\sx_{11} \dotsm \hat\sx_{1i} \dotsm \sx_{1N}}{ \sx_{21} \dotsm \hat\sx_{2j} \dotsm \sx_{2N} }.
			\label{eq_doubleexp}
		}
		The term $\kp\cp{\sw_1}{\sw_2}$ is precisely what was discussed in lemma \ref{lemma_fullpart} and
		\al[*]{
			(-1)^{i+j} \cp{\sx_{11} \dotsm \hat\sx_{1i} \dotsm \sx_{1N}}{ \sx_{21} \dotsm \hat\sx_{2j} \dotsm \sx_{2N} }
		}
		with $i=1$ was the coefficient of $\cp{\sx_{1i}}{\sx_{2j}}$ in its expansion. Hence, repeating the steps from that proof we immediately find the result for $|\cE|=1$:
		\al{\nonumber
			\sum_{ (\su_{\id},\sv_{\id})\in \sP_2 }\nquad \Big( \kp\cp{\su_{\id}\sy}{\sv_{\id}\sy} &- \kp[\Gamma^{\bullet}] \cp{\su_{\id}}{\sv_{\id}} \Big)\\
				& = (-\kp)^{1-N}\!\!\! \sum_{D\in\cD_{\Gamma}^1} (-2)^{c_2^2(D)} \Big(\!\! \prod_{(u,v)\in E_D^0}\nquad \cp{\sa_u}{\sa_v}\Big)\!\! \prod_{w\in V_D^{(2)}}\nquad x_{\Gamma}^w.
				\label{eq_z1_lemma_E1}
		}
		Replacing the Dodgson polynomials with $x_{\Gamma}^w$ (see \refeq{eq_xx_dp_signs}) flips the sign in front of the sum in \refeq{eq_doubleexp}. Since it is a double sum we get a factor of $2$. This, together with a $-1$ due to the factor $+\kp$ on the l.h.s. raises the power of $-2$ to $\tilde c(D)$. This can be interpreted as due to the additional tricoloured cycle that all diagrams $D\in \cD_{\Gamma}^1$ have.
			
		Now we can simply repeat the arguments of lemma \ref{lemma_part_2} to extend this to $|\cE|>1$ to finish the proof. Inclusion of the factor
		\al{
			\sum_{ (\su',\sv')\in\lambda_{\cE}(\su,\sv) } \frac{\cp{\su'\sy}{\sv'\sy}}{\cp{\su'}{\sv'}}
		}
		simply turns each summand into a sum of $|\cE|$ terms where one of the factors in each of them is replaced with the Dodgson polynomials with appended letters $\sy$. Expanding that factor as above yields the term $\cp{\sy}{\sy} = \kp[\Gamma^{\bullet}]$ in each summand, so we need $-|\cE|\kp[\Gamma^{\bullet}]$ to cancel it. The remaining terms can then be collected into groups of terms that either already factorise or can be reduced with the exact same arguments as in lemma \ref{lemma_part_2}.
	\end{proof}

	% EXAMPLE 
	\subsubsection{An example}
		
	Before we move on to prove the main theorem we discuss an example to illustrate the previous lemma.\\
	
	Consider a sum over word pairs $(\su_{\id}, \sv_{\id}) \in \sP_j$ as before, but add in each Dodgson polynomial an additional letter $\sy$ representing a vertex. Due to this additional letter we constrain ourselves to an $N=2$ example, namely $n=(2)$. The word pairs are then $(\sa_1\sa_3, \sa_2\sa_4)$ and $(\sa_1\sa_4, \sa_2\sa_3)$, where we choose the colour $j$ to be that of the edges $(1,2)$ and $(3,4)$, and we expand the Dodgson polynomials as
	\al{\nonumber
			\kp^2\big(  \cp{\sa_1\sa_3\sy}{\sa_2\sa_4\sy} &+ \cp{\sa_1\sa_4\sy}{\sa_2\sa_3\sy} \big)\\[3mm]\nonumber
				=& 2 \cp{\sa_1}{\sa_2} \cp{\sa_3}{\sa_4} \cp{\sy}{\sy} - 2 \cp{\sa_1}{\sa_2} \cp{\sa_3}{\sy} \cp{\sa_4}{\sy} - 2 \cp{\sa_3}{\sa_4} \cp{\sa_1}{\sy} \cp{\sa_2}{\sy}\\\nonumber
				& - \cp{\sa_1}{\sa_3} \cp{\sa_2}{\sa_4} \cp{\sy}{\sy} + \cp{\sa_1}{\sa_3} \cp{\sa_2}{\sy} \cp{\sa_4}{\sy} + \cp{\sa_2}{\sa_4} \cp{\sa_1}{\sy} \cp{\sa_3}{\sy}\\
				& - \cp{\sa_1}{\sa_4} \cp{\sa_2}{\sa_3} \cp{\sy}{\sy} + \cp{\sa_1}{\sa_4} \cp{\sa_2}{\sy} \cp{\sa_3}{\sy} + \cp{\sa_2}{\sa_3} \cp{\sa_1}{\sy} \cp{\sa_4}{\sy}.
	}
	Note that there are $9$ distinct terms. Firstly, we have the $3=(2N-1)!!$ terms
	\al{
		\cp{\sy}{\sy} \big( 2 \cp{\sa_1}{\sa_2} \cp{\sa_3}{\sa_4} - \cp{\sa_1}{\sa_3} \cp{\sa_2}{\sa_4} - \cp{\sa_1}{\sa_4} \cp{\sa_2}{\sa_3} \big)
	}
	corresponding to diagrams $D\in \cD^2_2$ ($\simeq \cD_{\Gamma}^0$ for some suitable $\Gamma$). Dividing by $\kp$ one finds that this exactly agrees with the sum predicted in lemma \ref{lemma_fullpart} but with a factor $\cp{\sy}{\sy}$.
	
	The other $6=N(2N-1)!!$ terms are
	\al{\nonumber
			 - 2 \big( \cp{\sa_1}{\sa_2} \cp{\sa_3}{\sy} \cp{\sa_4}{\sy} + \cp{\sa_3}{\sa_4} \cp{\sa_1}{\sy} \cp{\sa_2}{\sy}\big)
				 & + \cp{\sa_1}{\sa_3} \cp{\sa_2}{\sy} \cp{\sa_4}{\sy} + \cp{\sa_2}{\sa_4} \cp{\sa_1}{\sy} \cp{\sa_3}{\sy}\\[2mm]
				 & + \cp{\sa_1}{\sa_4} \cp{\sa_2}{\sy} \cp{\sa_3}{\sy} + \cp{\sa_2}{\sa_3} \cp{\sa_1}{\sy} \cp{\sa_4}{\sy}.
	}
	and correspond to diagrams $D\in \cD^2_1$ with one missing chord. Alternatively we can write this as
	\al{\nonumber
			2 \big(\cp{\sa_1}{\sa_2} x_{\Gamma}^3x_{\Gamma}^4 + \cp{\sa_3}{\sa_4} x_{\Gamma}^1x_{\Gamma}^2 \big)
				 & - \cp{\sa_1}{\sa_3} x_{\Gamma}^2x_{\Gamma}^4 - \cp{\sa_2}{\sa_4} x_{\Gamma}^1x_{\Gamma}^3\\[2mm]
				 & - \cp{\sa_1}{\sa_4} x_{\Gamma}^2x_{\Gamma}^3 - \cp{\sa_2}{\sa_3} x_{\Gamma}^1x_{\Gamma}^4,
	}
	and we see that the factors are as predicted by lemma \ref{lemma_z1_lemma}, specifically the $|\cE|=1$ case in \refeq{eq_z1_lemma_E1}.\\

	We continue the example to a partition with two parts. Since we chose the colour of the word pairs to be that of $(1,2)$ and $(3,4)$, the partition needs to be of the other edges, i.e. $\cE = \{ \{(1,4)\}, \{(2,3)\} \}$. One finds
	\al{\nonumber
			-\kp \big( & \cp{\sa_1}{\sa_4}\cp{\sa_2\sy}{\sa_3\sy} + \cp{\sa_2}{\sa_3}\cp{\sa_1\sy}{\sa_4\sy} \big)\\[3mm]\nonumber
					 &= - 2 \cp{\sa_1}{\sa_4} \cp{\sa_2}{\sa_3} \cp{\sy}{\sy}
					 			+ \cp{\sa_3}{\sy} \cp{\sa_4}{\sy} + \cp{\sa_3}{\sa_4}  \cp{\sa_1}{\sy} \cp{\sa_2}{\sy},\\[3mm]
					& = - 2 \cp{\sa_1}{\sa_4} \cp{\sa_2}{\sa_3} \kp[\Gamma^{\bullet}]
								 - \cp{\sa_1}{\sa_4}  x_{\Gamma}^2 x_{\Gamma}^3 - \cp{\sa_2}{\sa_3}x_{\Gamma}^1 x_{\Gamma}^4.
	}
	The above results from the pairs $\lambda_{\cE}(\sa_1\sa_3,\sa_2\sa_4) = \{ (\sa_1,\sa_4), (\sa_3,\sa_2)\}$. Note that this also yields a sign $\sgn_{\cE}(\sa_1\sa_3,\sa_2\sa_4) =-1$ in front of $\kp$ on the l.h.s. since $\sa_4$ and $\sa_2$ are permuted when concatenating the two word pairs in $\lambda_{\cE}(\sa_1\sa_3,\sa_2\sa_4)$. The other word pair yields $\lambda_{\cE}(\sa_1\sa_4,\sa_2\sa_3) = \emptyset$, and thus no polynomial. We see that only two diagrams are in $\cD|_{\cE}^1$, since $(1,4)$ and $(2,3)$ are the only two possible chords that stay within one part of the partition $\cE = \{ \{(1,4)\}, \{(2,3)\} \}$. For the other term note that the $-2$ does not come from the number of cycles but from $|\cE|=2$ together with the signum.

	% THEOREM
	\paragraph{Proof of Theorem \ref{theo_main_2}.}
	The partition polynomial definitions \ref{def_PartPol} and \ref{def_PartPol_1} together with lemma \ref{lemma_z1_lemma} directly yield
	\al{\nonumber
		Z_{\Gamma}^1
			&=
			\frac{1}{2}\sum_{\cE \in \cP(E_D^1)} (-\kp)^{N-|\cE|} (|\cE|+1)! \\ 
			& \qquad \times \sum_{ (\su,\sv)\in \sP_2 } \sgn_{\cE}(\su,\sv) \nquad \prod_{ (\su',\sv')\in\lambda_{\cE}(\su,\sv)}\nqquad\cp{\su'}{\sv'}
			\bigg( |\cE|\vssp - \kp \nqquad \sum_{ (\su',\sv')\in\lambda_{\cE}(\su,\sv) } \frac{\cp{\su'\sy}{\sv'\sy}}{\cp{\su'}{\sv'}} \bigg)\\[3mm]
			& = \frac{1}{2}\sum_{\cE \in \cP(E_D^1)} (-1)^{|\cE|+1} (|\cE|+1)!  \sum_{D \in \cD|_{\cE}^1} (-2)^{c_2^2(D)} 
					\Big(\nquad\prod_{(u,v)\in E_D^0} \nquad\cp{u}{v}\Big)\!\! \prod_{w\in V_D^{(2)}}\nquad x_{\Gamma}^w.
	}
	Now we have almost the same situation as in theorem \ref{theo_main}, except for the summation over $\cD|_{\cE}^1$ instead of $\cD|_{\cE}^0$. We can again exploit the one-to-one correspondence between diagrams in $\cD|_{\cE}^0$ and subsets of $N$ diagrams in $\cD|_{\cE}^1$ to be able to use the same argument as before. This correspondence carries over to the restricted sets and we can split the sum over $\cD|_{\cE}^1$ into a sum over $\cD|_{\cE}^0$ and the chords of each diagram (see also \refeq{eq_sumD0D1}). We then have
	\al{\nonumber
		Z_{\Gamma}^1 
			& = \frac{1}{2} \sum_{\cE \in \cP(E_D^1)} (-1)^{|\cE|+1} (|\cE|+1)!  \sum_{D\in \cD|_{\cE}^0} (-2)^{c_2^2(D)-1} 
					\Big(\nquad\prod_{(u,v)\in E_D^0} \nquad\cp{u}{v}\Big) \sum_{(u,v)\in E_D^0} \frac{x_{\Gamma}^ux_{\Gamma}^v}{\cp{u}{v}}.
	}	
	Exchange of summations yields the same sum involving the Stirling numbers of the second kind, which allows us to find $c_2^1(D)$. We now have	
	\al{\nonumber
		Z_{\Gamma}^1
			&= \frac{1}{2} \sum_{D\in \cD_{\Gamma}^0} (-2)^{\tilde c(D)-1} \Big( \prod_{(u,v)\in E_D^0}\nquad \cp{u}{v} \Big) \sum_{(u,v)\in E_D^0}
												 \frac{x_{\Gamma}^ux_{\Gamma}^v}{\cp{u}{v}}\\[2mm]
			&= \frac{1}{2} \sum_{D\in \cD_{\Gamma}^1} (-2)^{\tilde c(D)} \bigg(\prod_{ (u,v)\in E_D^0 } \cp{u}{v}\bigg) \prod_{ w\in V_D^{(2)} } x_{\Gamma}^w.
	}
	where we take care to account for the reduced cycle number when translating back to a sum over $\cD_{\Gamma}^1$. This is exactly \refeq{eq_maintheo_2} and the proof is done.
	\qed

%4
% APPLICATION
% Application
\section{Application to Feynman integrals}\label{sec_application}

	\subsection{Structure of the integrand}\label{subsec_structure}

	We now return to Feynman integrals and apply the theorems we just proved. In order to do so we first need to combine the results of \cite{Golz_2017_CyclePol} and \cite{Golz_2017_Traces}. Our starting point is 	the unrenormalised integral
	\al[*]{
		\phi_{\Gamma} = \int_{\RR_+^{|E_{\Gamma}|}} \rd\alpha_1\dotsm \rd\alpha_{E_{\Gamma}} \frac{\exp\Big( \frac{\ssp}{\kp} \Big)}{ \kp^{2+h_1(\Gamma)} } \sum_{k=0}^{h_1(\Gamma)} \frac{ I_{\Gamma}^{(k)} }{\kp^k}
	}
	from \refeq{eq_int_unren}. It is convenient to consider the $k$-th summand contracted with a metric tensor corresponding to the two external vertices $x,y\in V_{\Gamma}$. Then
	\al{
		g_{\mu_x\mu_y}I_{\Gamma}^{(k)} &= -\tr(\gamma^{\mu_1}\dotsm\gamma^{\mu_{2h_1}}) \bigg( \prod_{ (u,v)\in E_{D_{\Gamma}}^0 }\nquad g_{\mu_u\mu_v} \bigg)
							\sum_{D\in \cD_{\Gamma}^k} \bigg( \prod_{ (u,v)\in E_D^0 } \frac{g_{\mu_u\mu_v}}{2}\cp{u}{v} \bigg) \prod_{w\in V_D^{(2)}} q_{\mu_w} x_{\Gamma}^w
	}
	is just a rewriting of the integrand as worked out in \cite[eq. (72)]{Golz_2017_CyclePol}. The sum over fermion edge subsets and pairings is interpreted in terms of chord diagrams whose vertices are labelled by fermion edges and the additional metric tensor adds a chord such that we indeed have sums over $\cD_{\Gamma}^0$ in $I_{\Gamma}^{(0)}$ etc.
	
	As all throughout this article we stay in the special case of photon propagator graphs, Feynman gauge, and quenched QED, which becomes manifest in the terms above as follows:
	\begin{itemize}
		\item A propagator graph has only two external vertices with a single external momentum $q$, such that one has the factorised polynomials $q_{\mu_w} x_{\Gamma}^w$.
		\item Because it is a photon propagator there is one closed fermion cycle, which leads to the trace of Dirac matrices. Since we have quenched QED there is only exactly one such cycle and therefore no product of traces. 
		\item For a general gauge each $I_{\Gamma}^{(k)}$ itself contains another sum
			\al{
				I_{\Gamma}^{(k)} = \sum_{l=0}^{h_1-1} \Big( \frac{\varepsilon}{\kp}\Big)^l  I_{\Gamma}^{(k,l)}	\label{eq_gengauge}
			}
			where the gauge parameter $\varepsilon$ is such that Feynman gauge is $\varepsilon\to 0$. Each $I_{\Gamma}^{(k,l)}$ is structurally similar to $I_{\Gamma}^{(k,0)}$ but instead of a skeleton chord diagram $D_{\Gamma}$ with $h_1$ fixed chords it contains a sum over all possible such skeletons with $h_1-l$ fixed chords and then the usual chord diagram sums over all possible additions of further chords. Note that this in particular means that the sum over $k$ then also goes up to $2h_1-1$. Here we stick to Feynman gauge and identify $I_{\Gamma}^{(k)} \equiv I_{\Gamma}^{(k)}|_{\varepsilon=0} = I_{\Gamma}^{(k,0)}$ for $0\leq k \leq h_1$.
	\end{itemize}

	\subsubsection{Contraction}\label{subsec_contract}
	
	Next we apply the contraction theorem \cite[Theorem 3.9]{Golz_2017_Traces} to remove all Dirac matrices and metric tensors. We find
	\al{
		g_{\mu_x\mu_y}I_{\Gamma}^{(k)}
					&= 2^{h_1+1} (-q^2)^k \sum_{D\in \cD_{\Gamma}^k} (-2)^{\tilde c(D)} \bigg( \prod_{ (u,v)\in E_D^0 } \cp{u}{v} \bigg) \prod_{w\in V_D^{(2)}} x_{\Gamma}^w.
	}	
	The final integer factor is computed as follows. There are a total of $2h_1-k$ chords yielding $(-2)^{2h_1-k}$, but the $h_1-k$ non-fixed chords added to $D_{\Gamma}^0$ come with a factor $1/2$. Free vertices, corresponding to Dirac matrices contracted with a momentum instead of a metric tensor yield powers of $q^2$  and there is one more factor of $-2$, due to the one base cycle of $D_{\Gamma}$, whose sign is cancelled by the $-1$ from the Feynman rules for a fermion cycle. Altogether one finds
	\al{
		\underbrace{-(-2)^1}_{ 1 \text{ base cycle/trace}}
				%\cdot \ \underbrace{ 2^{k-h_1} }_{ \text{from } \frac{1}{2}g_{\mu_u\mu_v} }
				\cdot \ \underbrace{ (-2)^{h_1} }_{ \substack{\text{fixed chords}\\ E_{D_{\Gamma}}^0 } }\ 
				\cdot \ \underbrace{ (-1)^{h_1-k} }_{ \substack{\text{other chords} \\ E_D^0 } } \ 
				\cdot \ \underbrace{ (q^2)^k }_{ \substack{ \text{free vertices} \\ V_D^{(2)} } } = 2^{h_1+1} (-q^2)^k.
	}
	
	For the actual integrand we are interested in $I_{\Gamma}^{(k)}$, not its contraction with $g_{\mu_x\mu_y}$, so we need to work out what the effect of this contraction is. To simplify notation, let $\mu$ and $\nu$, without subscript, denote the space time indices of the external vertices, previously written $\mu_x,\mu_y$. For $k=0$ there are no free vertices. In other words, the added chord between external vertices causes the contraction $g_{\mu\nu}g^{\mu\nu} = 4$ which is counteracted by a factor $2^{-2}$ for all $D\in \cD_{\Gamma}^0$. Hence,
	\al{
		I_{\Gamma}^{(0)} = g^{\mu\nu}2^{h_1-1} \sum_{D\in \cD_{\Gamma}^0} (-2)^{\tilde c(D)} \prod_{ (u,v)\in E_D^0 } \cp{u}{v} = g^{\mu\nu}2^{h_1} Z_{\Gamma}^0
	}
	with theorem \ref{theo_main}.\\
	
	For $k\geq 1$ the chord diagrams split into three disjoint subsets $\cD_{\Gamma,\bullet}^k \subset \cD_{\Gamma}^k$. The correction factor depends on the result of the Dirac matrix contraction without contraction of the external vertices. This, in turn, is characterised by $\sgn(x,y)$, the signum of the external vertices in the chord diagram (introduced in section \ref{sec_colours}):
	\al{
		\sgn(x,y)=
			\begin{cases}
				+1	 		\qquad\rightarrow \tr( \slashed q \slashed q q^{\mu}q^{\nu} )		& \sim q^2 g^{\mu\nu}\\[3mm]
				\hphantom{+} 0	\qquad\rightarrow \tr( \slashed q q^{\mu} \slashed q q^{\nu} )		&\sim 2q^{\mu}q^{\nu} - q^2g^{\mu\nu}\\[3mm] 
				-1	 		\qquad\rightarrow \tr( \slashed q q^{\mu})\tr( \slashed q q^{\nu} )	&\sim q^{\mu}q^{\nu}
			\end{cases}
	}
	Contracting the results on the r.h.s. with $g_{\mu\nu}$ one sees that the correction factor is a $-2$ with the exponent $1+\sgn(x,y)$ for all diagrams, including the $k=0$ case, in which only $+1$ occurs. We can define partial chord diagram sums
	\al{
		Z_{\Gamma,\bullet}^k \defeq \frac{1}{2}\sum_{D\in \cD_{\Gamma,\bullet}^k} (-2)^{\tilde c(D)} \bigg( \prod_{ (u,v)\in E_D^0 } \cp{u}{v} \bigg) \prod_{w\in V_D^{(2)}} x_{\Gamma}^w,
	}
	based on these subsets. Then $Z_{\Gamma,+}^0 = Z_{\Gamma}^0$ and $Z_{\Gamma,-}^0 = Z_{\Gamma,0}^0 = 0$, and for $k=1$ we have 
	\al{
		Z_{\Gamma}^1 = Z_{\Gamma,+}^1 + Z_{\Gamma,0}^1 + Z_{\Gamma,-}^1
	}
	with theorem \ref{theo_main_2}. For $k>1$ similar equalities should hold, assuming one defines the right $k$-th order partition polynomial, but, as we will see in the next section, for a superficially renormalised integral these two will suffice.
	
	With this notation the $k$-th summand now has become
	\al{\nonumber
		I_{\Gamma}^{(k)} &= 2^{h_1} (-q^2)^k \bigg( g^{\mu\nu} Z_{\Gamma,+}^k - 2\Big( 2\frac{q^{\mu}q^{\nu}}{q^2} - g^{\mu\nu} \Big) Z_{\Gamma,0}^k + 4\frac{q^{\mu}q^{\nu}}{q^2}Z_{\Gamma,-}^k \bigg)\\[3mm]
					&= 2^{h_1} (-q^2)^k \Big( g^{\mu\nu} \big( Z_{\Gamma,+}^k + 2Z_{\Gamma,0}^k \big)  + 4\frac{q^{\mu}q^{\nu}}{q^2}\big( Z_{\Gamma,-}^k - Z_{\Gamma,0}^k \big)  \Big),
	}
	and the full unrenormalised integrand is
	\al{
		I_{\Gamma} = 2^{h_1}\frac{e^{-\frac{\ssp}{\kp}}}{\kp^{h_1+2}} \Big(  g^{\mu\nu}Z_{\Gamma}^0 - \sum_{k=1}^{h_1} \frac{(-q^2)^{k-1}}{ \kp^k}
														\Big( q^2g^{\mu\nu} \big( Z_{\Gamma,+}^k + 2Z_{\Gamma,0}^k \big)  + 4q^{\mu}q^{\nu}\big( Z_{\Gamma,-}^k - Z_{\Gamma,0}^k \big)  \Big).
	}
	Now this may seem somewhat problematic -- we know the sum $Z_{\Gamma,+}^1 + Z_{\Gamma,0}^1 + Z_{\Gamma,-}^1$ (and can presumably generalise that knowledge to $k>1$). But what can we do about these combinations? As it turns out, we can exploit the transversality of the photon propagator to modify the integrand such that it only contains these types of sums, but first we want to renormalise it.

	\subsubsection{Renormalisation}
	
	We (superficially) renormalise this integrand in a BPHZ scheme following \cite{BrownKreimer_2013_AnglesScales}. Consider a generic integral of the same form as our Feynman integral, namely
	\al{
		\int_{\RR_+^n} \rd\alpha_1\dotsm\rd\alpha_n \sum_{k=-1}^m e^{-X}I_k, \label{eq_genint}
	}
	where $X$ and all $I_k$ are rational functions in $\alpha_i$ with overall degree (degree of numerator minus degree of denominator) $1$ and $k-n$ respectively.
	
	We can introduce an auxiliary variable $t$ by inserting $1= \int_0^{\infty} \delta(t-\sum_i \lambda_i \alpha_i) \rd t$, where each $\lambda_i\in \{0,1\}$ and at least one of them non-zero. Then scaling all Schwinger parameters by $\alpha_i \mapsto t\alpha_i$ turns \refeq{eq_genint} into	
	\al{
		\int_{\RR_+^n} \rd\alpha_1\dotsm\rd\alpha_n\ \delta(1-\sum_i \lambda_i\alpha_i) \sum_{k=-1}^m  I_k T_k \label{eq_delta}
	}	
	with
	\al{
		T_k =\int_0^{\infty}\rd t\  t^{k-1} e^{-tX} = X^{-k} \Gamma(k). \label{eq_gammafct}
	}
	The Gamma function has poles at negative integers and zero, corresponding here to quadratic and logarithmic divergences for $k=-1$ and $0$. They can be parametrised for further study by regularising the $t$-integration with an $\epsilon>0$:
	\al{	% \lim_{\epsilon\to 0}
		T_0 &\stackrel{\epsilon\to 0}{=} \int_{\epsilon}^{\infty} t^{-1} e^{-t X} \rd t	=  -\log \epsilon - \log X - \gamma_E + \mathcal O(\epsilon)\\[2mm]
		T_{-1} &\stackrel{\epsilon\to 0}{=} \int_{\epsilon}^{\infty} t^{-2} e^{-t X} \rd t = \frac{e^{-\epsilon X}}{\epsilon} -X \underbrace{\int_{\epsilon}^{\infty} t^{-1} e^{-tX} \rd t}_{=T_0} \label{eq_quaddiv}
	}
	We see that the divergent terms are isolated and a simple subtraction like
	\al{
		T_0 - T_0' = -\log \frac{X}{X'},\label{eq_simplesub}
	}
	is already enough to cancel a logarithmic divergence. The quadratic divergence requires first an on-shell subtraction to remove the term $\sim \epsilon^{-1}$, then the usual subtraction for the remaining logarithmic divergence.\\
	
	Note that, assuming convergence, the integral in \refeq{eq_delta} can equivalently be written projectively\footnote{ For a more thorough discussion of the bijection between $\RR_+^n$ and (a certain subset of) projective space induced by the introduction of the delta function see \cite[sec. 2.1.3]{Panzer_2015_PhD}. Of note in particular is the fact that it is completely independent of the choice of the parameters $\lambda_i$, which is sometimes called ``Cheng-Wu theorem''.}
	\al{
		\int_{\RR_+^n} \rd\alpha_1\dotsm\rd\alpha_n\ \delta(1-\sum_i \lambda_i\alpha_i) \sum_{k=-1}^m  I_k T_k
		 = \int_{\sigma_{\Gamma}} \Omega_{\Gamma} \sum_{k=-1}^m  I_k T_k,
	}
	where $\Omega_{\Gamma} = \sum_{i=1}^{n}(-1)^{i-1}\alpha_i\mathrm d\alpha_1\wedge\dotsm\wedge\widehat{\mathrm d\alpha_i}\wedge\dotsm\wedge\mathrm d\alpha_n$ and one integrates over the subset of real projective space in which all parameters are positive
	\al{
		\sigma_{\Gamma} = \{[\alpha_1: \dotso :\alpha_n]\ | \ \alpha_i > 0 \ \forall \ i=1,\dotsc, n \}.
	}
	For brevity we will use this notation from now on.\\

	We can now apply this to the integrand. Simply counting the degrees of the various homogenous polynomials that appear in numerator and denominator one finds that the $0$-th term is quadratically divergent, the next one logarithmically, and all others are convergent. Hence, the (superficially) renormalised integrand is
	\al{\nonumber
		I_{\Gamma}^R 	&= \underbrace{\log \frac{q^2}{\mu^2}}_{\bdefeq L}  
											\frac{2^{h_1}}{\kp^{h_1+3}} \Big(  q^2g^{\mu\nu} \vssp Z_{\Gamma}^0
													+  q^2g^{\mu\nu} \big( Z_{\Gamma,+}^1 + 2Z_{\Gamma,0}^1 \big)  + 4q^{\mu}q^{\nu}\big( Z_{\Gamma,-}^1 - Z_{\Gamma,0}^1 \big)  \Big)\\[3mm]
					&= \frac{2^{h_1} L}{\kp^{h_1+3}} \Big(  q^2g^{\mu\nu} \big( \vssp Z_{\Gamma}^0 + Z_{\Gamma,+}^1 + 2Z_{\Gamma,0}^1 \big)
						 						   - q^{\mu}q^{\nu}\big( 4Z_{\Gamma,0}^1 - 4Z_{\Gamma,-}^1 \big) \Big).
	}
	At this point we can now impose transversality on the integrand to simplify it. For the photon propagator transversality simply means that the amplitude, the sum of all relevant Feynman integrals, is proportional to $q^2g^{\mu\nu} - q^{\mu}q^{\nu}$. This is manifestly not true for individual Feynman integrals, let alone their integrands. However, since only their sum has physical meaning we can simplify redefine $I_{\Gamma}^R$ such that it already satisfies transversality. Whatever change this effects in the integral cancels when adding up all integrals. Here we get the condition
	\al{
		\vssp Z_{\Gamma}^0 + Z_{\Gamma,+}^1 + 2Z_{\Gamma,0}^1 \stackrel{!}{=} 4Z_{\Gamma,0}^1 - 4Z_{\Gamma,-}^1.	\label{eq_transvers_cond}
	}
	We could now naively just use either side of this in the integrand. However, we can also do better than that. Note that
	\al{
		\vssp Z_{\Gamma}^0 + Z_{\Gamma,+}^1 + 2Z_{\Gamma,0}^1
			= \big(\vssp Z_{\Gamma}^0 + \underbrace{Z_{\Gamma,+}^1 + Z_{\Gamma,0}^1 + Z_{\Gamma,-}^1}_{ =Z_{\Gamma}^1} \big) + \big( Z_{\Gamma,0}^1 - Z_{\Gamma,-}^1\big).
	}
	Now imposing the transversality condition \refeq{eq_transvers_cond} yields
	\al{
		\vssp Z_{\Gamma}^0  + Z_{\Gamma}^1 = 3\big( Z_{\Gamma,0}^1 - Z_{\Gamma,-}^1\big)
	}
	and the integrand becomes
	\al{\nonumber
		I_{\Gamma}^R 	&= \frac{2^{h_1} L}{\kp^{h_1+3}} \Big(  q^2g^{\mu\nu} \big( \vssp Z_{\Gamma}^0 + Z_{\Gamma,+}^1 + 2Z_{\Gamma,0}^1 \big)
						 						   - 4q^{\mu}q^{\nu}\big( Z_{\Gamma,0}^1 - Z_{\Gamma,-}^1 \big) \Big)\\[3mm]
					&= ( q^2g^{\mu\nu} - q^{\mu}q^{\nu} ) L \frac{ 2^{h_1+2} }{3} \frac{ \vssp Z_{\Gamma}^0  + Z_{\Gamma}^1 }{ \kp^{h_1+3} }.
	}
	We can also use the definitions of the partition polynomials to make the cancellations more obvious:
	\al{
		\frac{ \vssp Z_{\Gamma}^0  + Z_{\Gamma}^1 }{ \kp^{h_1+3} }
			 = \sum_{l=1}^{h_1} (-1)^{h_1-l} (l+1)! \bigg( \frac{\vssp Z_{\Gamma}^0\big|_l}{\kp^{l+3}} - \frac{1}{2}\frac{Z_{\Gamma}^1\big|_l}{\kp^{l+2}} \bigg)
	}

	%
	%
	%
	% EXAMPLES
	\subsection{Examples}
	
	% 1 LOOP
	\subsubsection{1-loop photon propagator}

	% FIGURE 1 loop photon
	\begin{figure}[h]\begin{center}
		\includegraphics[scale=1.0]{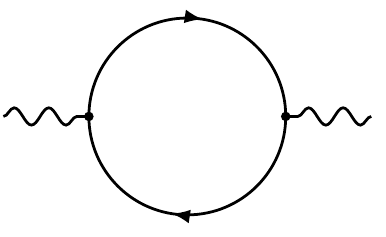}
		\caption[The 1-loop photon propagator]{The 1-loop photon propagator.}
		\label{06_img_1P}
	\end{center}\end{figure}
	
	The 1-loop case is the only primitive photon propagator and therefore the only example we can show in full without discussing subdivergences. The Kirchhoff polynomial is $\kp = \alpha_1+\alpha_2$ and $\vssp = \alpha_1\alpha_2$. The only possible cycle polynomial $\cp{1}{2}$ and $Z_{\Gamma}^0\big|_1$ are both just $1$ and $Z_{\Gamma}^1\big|_1 = -\vssp/\kp$.
	The integrand is therefore
	\al{
		\frac{ \vssp Z_{\Gamma}^0  + Z_{\Gamma}^1 }{ \kp^{h_1+3} }
			= (-1)^{1-1} (1+1)! \bigg( \frac{\vssp}{\kp^{4}} + \frac{1}{2}\frac{\vssp}{\kp^{4}} \bigg)
			= 3\frac{\vssp}{\kp^{4}}
			= 3\frac{ \alpha_1\alpha_2}{ (\alpha_1+\alpha_2)^4}
	}
	and the renormalised integral is
	\al{\nonumber
		\phi_{\Gamma}^R &= (q^2g^{\mu\nu} - q^{\mu}q^{\nu}) \frac{2^{h_1+2}}{3} L \int_{\sigma_{\Gamma}} \Omega_{\Gamma} 
											\frac{ \vssp Z_{\Gamma}^0 + Z_{\Gamma}^1 }{\kp^{h_1+3}}\\[3mm]\nonumber
					 &= 8L(q^2g^{\mu\nu} - q^{\mu}q^{\nu}) \int_{\sigma_{\Gamma}} \frac{ \alpha_1\alpha_2}{ (\alpha_1+\alpha_2)^{4}}\Omega_{\Gamma}\\
						&= \frac{4}{3}L (q^2g^{\mu\nu} - q^{\mu}q^{\nu}).
	}
	The factor $4/3$ is exactly the 1-loop coefficient of the QED beta function in the conventions of \cite{BroadhurstDelbourgoKreimer_1996_Unknotting, GKLS_1991_QED, Rosner_1966_QED}.

	% 3 LOOPS
	\subsubsection{3-loop photon propagators}
	
	For Feynman graphs with more than one loop we can not compute the full integral without discussing subdivergences and including the corresponding terms of Zimmermann's forest formula for a fully renormalised integrand. However, we can show what the superficially renormalised part of the integrand looks like and especially emphasise the cancellations and reductions in size due to the two summation theorems.
			
	\begin{figure}[H]\begin{center}
		\begin{subfigure}{0.2\textwidth}\begin{center}
			\includegraphics[height=0.1\textheight]{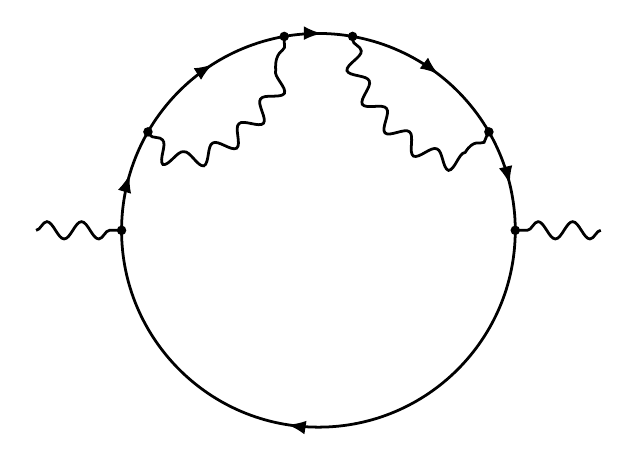}
			\caption{}
			\label{06a_img_3_1}
		\end{center}\end{subfigure}
		\begin{subfigure}{0.2\textwidth}\begin{center}
			\includegraphics[height=0.1\textheight]{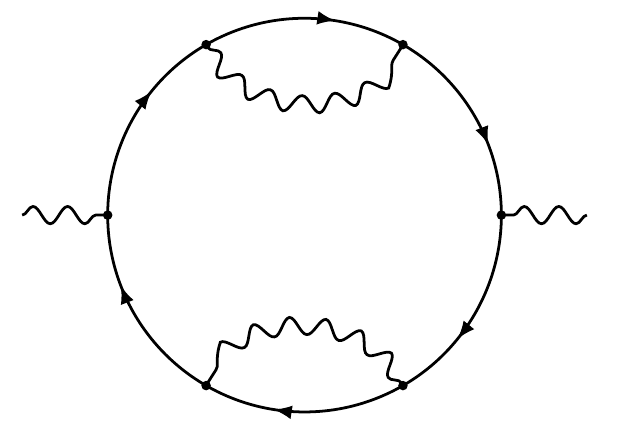}
			\caption{}
			\label{06a_img_3_2}
		\end{center}\end{subfigure}
		\begin{subfigure}{0.2\textwidth}\begin{center}
			\includegraphics[height=0.1\textheight]{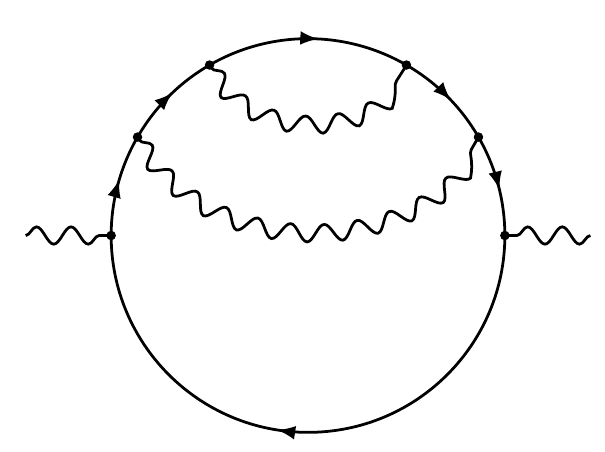}
			\caption{}
			\label{06a_img_3_3}
		\end{center}\end{subfigure}
		\begin{subfigure}{0.2\textwidth}\begin{center}
			\includegraphics[height=0.1\textheight]{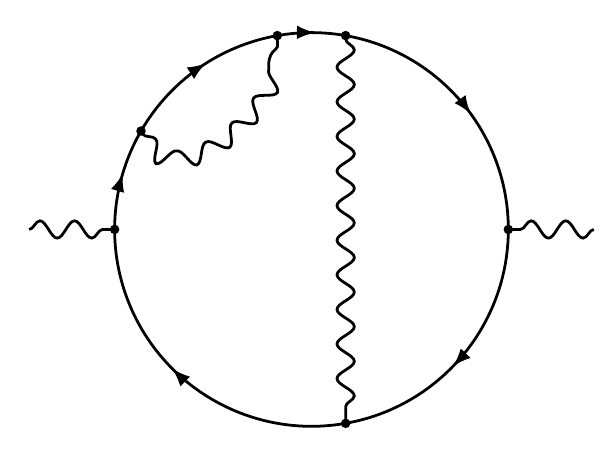}
			\caption{}
			\label{06a_img_3_4}
		\end{center}\end{subfigure}
		\begin{subfigure}{0.2\textwidth}\begin{center}
			\includegraphics[height=0.1\textheight]{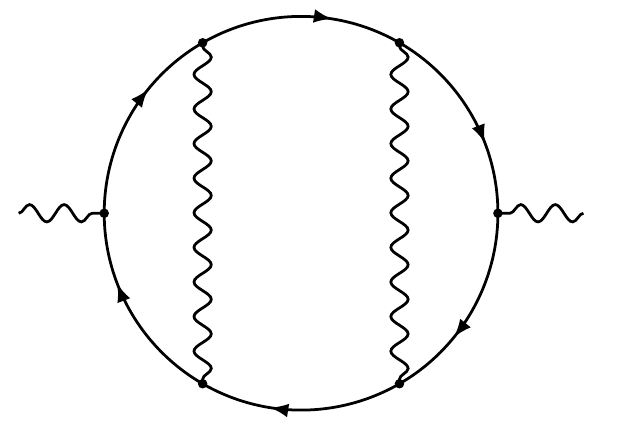}
			\caption{}
			\label{06a_img_3_5}
		\end{center}\end{subfigure}
		\begin{subfigure}{0.2\textwidth}\begin{center}
			\includegraphics[height=0.1\textheight]{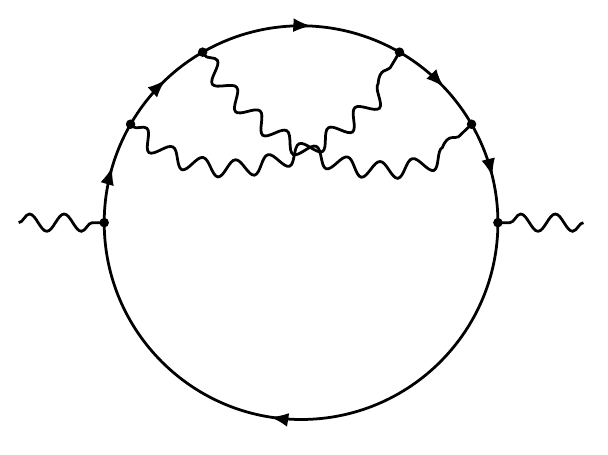}
			\caption{}
			\label{06a_img_3_6}
		\end{center}\end{subfigure}
		\begin{subfigure}{0.2\textwidth}\begin{center}
			\includegraphics[height=0.1\textheight]{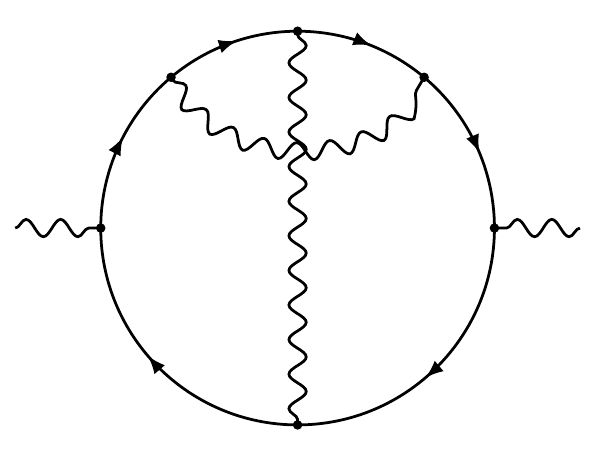}
			\caption{}
			\label{06a_img_3_7}
		\end{center}\end{subfigure}
		\begin{subfigure}{0.2\textwidth}\begin{center}
			\includegraphics[height=0.1\textheight]{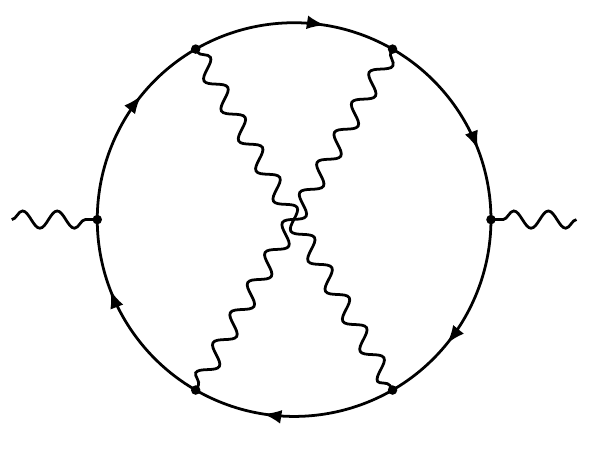}
			\caption{}
			\label{06a_img_3_8}
		\end{center}\end{subfigure}
		\caption{The 3-loop topologies with one fermion cycle.}
	\end{center}\end{figure}

	At two loops the examples are still rather simple so we go to three loops, where the integrals start to become much more involved. For example, $Z_{\Gamma}^0$ is now already a polynomial of degree $h_1(h_1-1) = 6$, compared to just $2$ at two loops, and the number of chord diagrams rises to $15$ (in Feynman gauge, and already hundreds in general gauge) such that the reduction to $h_1=3$ small summands in the partition polynomial now becomes significant. All examples were computed with Maple \footnote{Maple\texttrademark\ is a trademark of Waterloo Maple Inc. \cite{Maple}.}.
	
	We focus on the graph in \reffig{06a_img_3_8}.	
	\begin{figure}\begin{center}
		\includegraphics[width=0.9\textwidth]{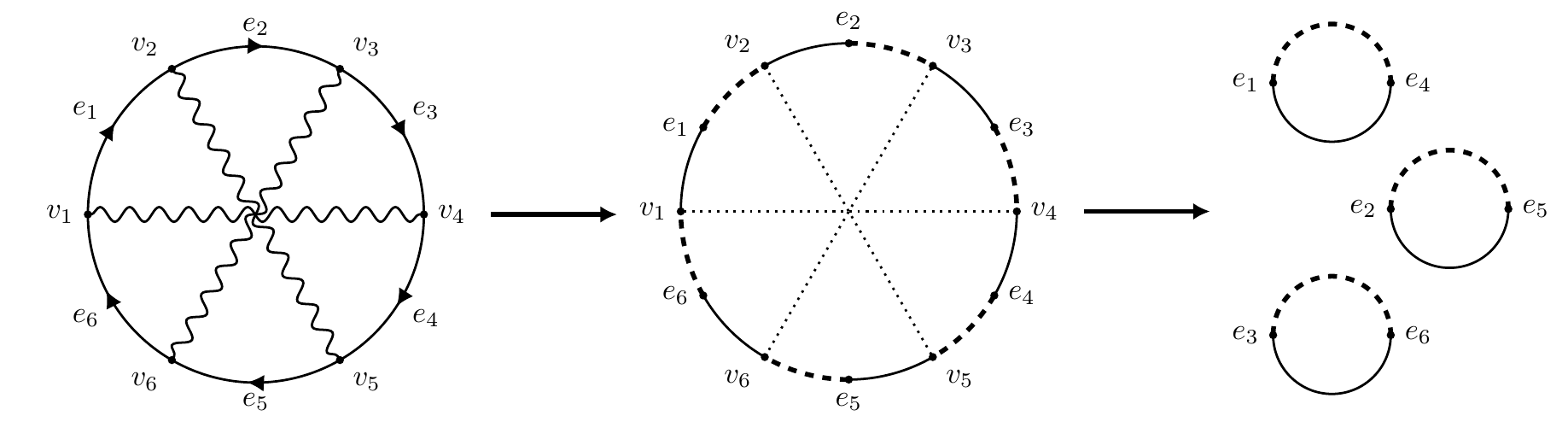}
		\caption{From left to right: The graph $\Gamma$ from \reffig{06a_img_3_8} with its external photon edge closed - the corresponding chord diagram $D_{\Gamma}$ with fixed chords corresponding to all photon edges - the projection $D_{\Gamma}^0 = \pi_0(D_{\Gamma})$.}
		\label{img_38_proj}
	\end{center}\end{figure}
	Label edges and vertices as in \reffig{img_38_proj} with $v_1,v_4$ being the external vertices and $e_7=(v_2,v_5)$ and $e_8 = (v_3,v_6)$ the two photon edges. The Kirchhoff and second Symanzik polynomial consist of $36$ and $45$ monomials, so we refrain from writing them out in full here. An example for a cycle polynomial is:
	\al{
		\cp{1}{6} = \alpha_2(\alpha_3 + \alpha_4 + \alpha_5 + \alpha_8) +  (\alpha_3+\alpha_4)(\alpha_5 + \alpha_7 + \alpha_8) + \alpha_7(\alpha_5+\alpha_8)
	}
	This one has so many terms since $e_1$ and $e_6$ share all their cycles, because they are incident to the same external (i.e. 2-valent) vertex of $\Gamma$. In other words, $\cp{1}{6} = \cp{1}{1} = \cp{6}{6}$. Others are simpler:
	\al{
		\cp{1}{3} = \cp{1}{4} = \cp{3}{6} = \cp{4}{6} = -\alpha_2\alpha_5 + \alpha_7\alpha_8
	}
	Here we have an example of monomials with different signs, which is due to the fact that the two corresponding cycles are twisted relative to each other (as discussed in the proof of proposition \ref{05a_prop_eqpol}). One cycle is the fermion cycle, the other crosses via both photon edges.

	% partition pol
	\paragraph{The partition polynomial $Z_{\Gamma}^0$.}
	The word pairs we get from $D_{\Gamma}^0$ are
	\al[*]{
		 ( \sa_1\sa_2\sa_3, \sa_4\sa_5\sa_6 ), \quad ( \sa_1\sa_5\sa_3, \sa_4\sa_2\sa_6 ),\quad
		 ( \sa_1\sa_2\sa_6, \sa_4\sa_5\sa_3 ), \quad ( \sa_1\sa_5\sa_6, \sa_4\sa_2\sa_3 ).
	}
	For $Z_{\Gamma}^0\big|_1$ we have the single partition $\cE = \{ \{ (e_1, e_4), (e_2, e_5), (e_3, e_6) \} \}$, such that
	\al{\nonumber
		Z_{\Gamma}^0\big|_1 &= \cp{\sa_1\sa_2\sa_3}{\sa_4\sa_5\sa_6} + \cp{\sa_1\sa_2\sa_6}{\sa_4\sa_5\sa_3} 
							+ \cp{\sa_1\sa_5\sa_3}{\sa_4\sa_2\sa_6} + \cp{\sa_1\sa_5\sa_6}{\sa_4\sa_2\sa_3}\\[3mm]
						&= 1+ 0 + 0 + 1 = 2.
						\label{eq_Z01}
	}
	% remark
	\begin{remark}
		Note that the two vanishing Dodgson polynomials are those that have the letter pairs $\sa_1/\sa_6$ and $\sa_3/\sa_4$ within the same word. While these are different letters we have seen above that their associated edges are equivalent as far as the cycle space of $\Gamma$ is concerned. This is reflected in the behaviour of the Dodgson polynomials, which vanish as if the letters were identical.
	\end{remark}

	For $Z_{\Gamma}^0\big|_2$ we have three partitions with two parts, consisting of one and two edges respectively. For $\cE_1 = \{ \{ (e_1, e_4) \}, \{ (e_2, e_5), (e_3, e_6) \} \}$ one has
	\al[*]{
		\lambda_{\cE_1}( \sa_1\sa_2\sa_3, \sa_4\sa_5\sa_6 ) &= \{ (\sa_1, \sa_4), (\sa_2\sa_3, \sa_5\sa_6) \},\\[2mm]
		\lambda_{\cE_1}( \sa_1\sa_2\sa_6, \sa_4\sa_5\sa_3 ) &= \{ (\sa_1, \sa_4), (\sa_2\sa_6, \sa_5\sa_3) \},\\[2mm]
		\lambda_{\cE_1}( \sa_1\sa_5\sa_3, \sa_4\sa_2\sa_6 ) &= \{ (\sa_1, \sa_4), (\sa_5\sa_3, \sa_2\sa_6) \},\\[2mm]
		\lambda_{\cE_1}( \sa_1\sa_5\sa_6, \sa_4\sa_2\sa_3 ) &= \{ (\sa_1, \sa_4), (\sa_5\sa_6, \sa_2\sa_3) \}.
	}
	All permutations give positive signs and the corresponding polynomial is
	\al{
		2\cp{\sa_1}{\sa_4} \big( \cp{\sa_2\sa_3}{\sa_5\sa_6} + \cp{\sa_2\sa_6}{\sa_5\sa_3} \big) = -2( -\alpha_2\alpha_5 + \alpha_7\alpha_8 )( \alpha_7 + \alpha_8 ),
	}
	which is also the polynomial one finds analogously for $\cE_3 = \{ \{ (e_1, e_4), (e_2, e_5) \},$ $\{ (e_3, e_6) \} \}$. For the third, $\cE_2 = \{ \{ (e_1, e_4), (e_3, e_6) \}, \{ (e_2, e_5) \} \}$, the words are
	\al[*]{
		\lambda_{\cE_2}( \sa_1\sa_2\sa_3, \sa_4\sa_5\sa_6 ) &= \{ (\sa_2, \sa_5), (\sa_1\sa_3, \sa_4\sa_6) \},\\[2mm]
		\lambda_{\cE_2}( \sa_1\sa_2\sa_6, \sa_4\sa_5\sa_3 ) &= \{ (\sa_2, \sa_5), (\sa_1\sa_6, \sa_4\sa_3) \},\\[2mm]
		\lambda_{\cE_2}( \sa_1\sa_5\sa_3, \sa_4\sa_2\sa_6 ) &= \{ (\sa_5, \sa_2), (\sa_1\sa_3, \sa_4\sa_6) \},\\[2mm]
		\lambda_{\cE_2}( \sa_1\sa_5\sa_6, \sa_4\sa_2\sa_3 ) &= \{ (\sa_5, \sa_2), (\sa_1\sa_6, \sa_4\sa_3) \}.
	}
	Again, the total sign is always positive but this time with $\sgn(\sigma)=-1=\sgn(\sigma')$, where e.g. $\sa_2\sa_1\sa_3 = \sigma(\sa_1\sa_2\sa_3)$ and $\sa_5\sa_4\sa_6 = \sigma'(\sa_4\sa_5\sa_6)$. The polynomial is
	\al{\nonumber
		2\cp{\sa_2}{\sa_5} & \big( \cp{\sa_1\sa_3}{\sa_4\sa_6} + \cp{\sa_1\sa_6}{\sa_4\sa_3} \big)\\
				& = 2( -(\alpha_1+\alpha_6)(\alpha_3+\alpha_4) + \alpha_7\alpha_8 )( \alpha_2 + \alpha_5 + \alpha_7 + \alpha_8 ).
	}
	The last polynomial is always of the same form. The only partition $\cE = \{ \{ (e_1, e_4)\},$ $\{ (e_2, e_5) \}, \{ (e_3, e_6) \} \}$ has each base edge in a separate part such that
	\al[*]{
		\lambda_{\cE}( \sa_1\sa_2\sa_3, \sa_4\sa_5\sa_6 ) &= \{ (\sa_1, \sa_4), (\sa_2, \sa_5), (\sa_3, \sa_6) \},\\[2mm]
		\lambda_{\cE}( \sa_1\sa_2\sa_6, \sa_4\sa_5\sa_3 ) &= \{ (\sa_1, \sa_4), (\sa_2, \sa_5), (\sa_6, \sa_3) \},\\[2mm]
		\lambda_{\cE}( \sa_1\sa_5\sa_3, \sa_4\sa_2\sa_6 ) &= \{ (\sa_1, \sa_4), (\sa_5, \sa_2), (\sa_3, \sa_6) \},\\[2mm]
		\lambda_{\cE}( \sa_1\sa_5\sa_6, \sa_4\sa_2\sa_3 ) &= \{ (\sa_1, \sa_4), (\sa_5, \sa_2), (\sa_6, \sa_3) \}.
	}
	and
	\al{\nonumber
		Z_{\Gamma}^0\big|_3 &= 4\cp{\sa_1}{\sa_4}\cp{\sa_2}{\sa_5}\cp{\sa_3}{\sa_6}\\
						&= ( -\alpha_2\alpha_5 + \alpha_7\alpha_8 )^2( -(\alpha_1+\alpha_6)(\alpha_3+\alpha_4) + \alpha_7\alpha_8 ).
	}
	\begin{remark}
		The fact that the $\lambda_{\cE}$ are all non-vanishing and often the same is due to the fact that $D_{\Gamma}^0$ has the base cycle structure $n=(1,1,1)$ (see \reffig{img_38_proj}). In this case the partitions $\cP(E_D^1)$ and words $\sP_2$ are in a sense maximally compatible, since each $1$-coloured base edge has a $2$-coloured base edge partner between the exact same vertices. In other words, $\cP(E_D^1) = \cP(E_D^2)$ and $\sP_1 = \sP_2$.
	\end{remark}
	
	The number of terms in the three polynomials $Z_{\Gamma}^0\big|_k$ is $1$, $22$ and $15$ respectively. For comparison, the full chord diagram sum consists of 437 monomials. Here we especially also see how hidden the factorisation of the Kirchhoff polynomials can be. The expression $a\kp^2Z_{\Gamma}^0\big|_1 + b\kp Z_{\Gamma}^0\big|_2 + cZ_{\Gamma}^0\big|_3$ should have $36^2\cdot 1 + 36\cdot 22 + 15 = 2103$ terms. But the same monomials may of course occur in different parts and add up or cancel to yield the 437 that are left in the sum, obscuring the pattern. See also table \ref{tab_z0} for the reduction observed for other graphs.

% table Z0
\begin{table}[h]
\centering
\begin{tabular}{|l|c|c|c|c|c|c|c|c|}
\hline
\ 					& (a)	& (b)	& (c)	& (d)	& (e)		& (f)		& (g)		& (h)		\vphantom{\Big|} \\[2mm]\hline
$\#Z_{\Gamma}^0$		& 9	& 9	& 44	& 84	& 348	& 231	& 448	& 437	\vphantom{\bigg|} \\[2mm]
$\#Z_{\Gamma}^0/\kp^6$	& 9	& 9	& 44	& 15	& 16 		& 72  	& 53  	& 38  	\vphantom{\bigg|} \\\hline
\end{tabular}
\caption{Number of terms in $Z_{\Gamma}^0$ compared to $Z_{\Gamma}^0/\kp^6$ after cancellations for all graphs from \reffig{06a_img_3_1} to \ref{06a_img_3_8}.}
\label{tab_z0}
\end{table}

\vspace{-5mm}
	% partition pol
	\paragraph{The partition polynomial $Z_{\Gamma}^1$.}
	
	We do not need to repeat the discussion of partitions etc. but can simply sum over all possible ways to append a letter to the word pairs. For $Z_{\Gamma}^1\big|_1$ this means we take $Z_{\Gamma}^0\big|_1$ from \refeq{eq_Z01} and get
	\al{\nonumber
		Z_{\Gamma}^1\big|_1 + \frac{\vssp}{\kp}Z_{\Gamma}^0\big|_1 
							=  \cp{\sa_1\sa_2\sa_3\sy}{\sa_4\sa_5\sa_6\sy} + \cp{\sa_1\sa_2\sa_6\sy}{\sa_4\sa_5\sa_3\sy} 
							+ \cp{\sa_1\sa_5\sa_3\sy}{\sa_4\sa_2\sa_6\sy} + \cp{\sa_1\sa_5\sa_6\sy}{\sa_4\sa_2\sa_3\sy}.
	}
	This already has 92 terms, so explicitly giving it here in terms of Schwinger parameters would not be particularly enlightening. Similarly one finds e.g.
	\al{\nonumber
		Z_{\Gamma}^1\big|_3 &+ 3\frac{\vssp}{\kp}Z_{\Gamma}^0\big|_3\\
						&=  4\cp{\sa_1\sy}{\sa_4\sy}\cp{\sa_2}{\sa_5}\cp{\sa_3}{\sa_6}
						+ 4\cp{\sa_1}{\sa_4}\cp{\sa_2\sy}{\sa_5\sy}\cp{\sa_3}{\sa_6}
						+ 4\cp{\sa_1}{\sa_4}\cp{\sa_2}{\sa_5}\cp{\sa_3\sy}{\sa_6\sy},
	}
	which has 1551 terms. We see that these expressions are still quite large, but they nonetheless represent a massive reduction in size compared to the full integrand without cancellations. In  table \ref{tab_total} the superficially renormalised integrands with and without cancellations are compared and one finds a reduction by roughly one order of magnitude.
	
% table Z0+Z1
\begin{table}[h]
\centering
\begin{tabular}{|l|c|c|c|c|c|c|c|c|}
\hline
\ 										& (a)		& (b)		& (c)		& (d)		& (e)		& (f)		& (g)		& (h) 	\vphantom{\Big|} \\[2mm]\hline
$\#(\vssp Z_{\Gamma}^0+Z_{\Gamma}^1)$		& 528	& 681	& 1937	& 4698	& 17641	& 8210	& 22627	& 25575  	\vphantom{\bigg|} \\[2mm]
$\#(\vssp Z_{\Gamma}^0+Z_{\Gamma}^1)/\kp^6$	& 88		& 329	& 387	& 513	& 1106 	& 782  	& 1637  	& 2439  	\vphantom{\bigg|} \\\hline
\end{tabular}
\caption{Total number of terms in the superficially renormalised integrand with and without cancellations for all graphs from \reffig{06a_img_3_1} to \ref{06a_img_3_8}.}
\label{tab_total}
\end{table}

% References
	\newpage
	\bibliographystyle{plainurl}
	\bibliography{bibliography}

\end{document}